\documentclass[a4paper,UKenglish,cleveref, autoref, thm-restate]{lipics-v2021}

\usepackage{enumerate}
\usepackage{amsmath}
\usepackage{todonotes}
\usepackage{xspace}
\usepackage{pdfpages}
\usepackage{algpseudocode}
\usepackage{algorithm}

\newboolean{shortver}
\setboolean{shortver}{false}% for long version

\newcommand{\EDS}{{\sc Edge Dominating Set}\xspace}
\newcommand{\PfiveFEVS}{\texorpdfstring{\ensuremath{P_5}{\sc-free-EVS}}{P5-free EVS}\xspace}
\newcommand{\PfiveFVS}{\texorpdfstring{\ensuremath{P_5}{\sc-free-VS}}{P5-free-VS}\xspace}
\newcommand{\DSEPVC}{{\sc  Double-Star Edge Partition with Vertex Cost}\xspace}

\newcommand{\Fd}{\mathcal{F}_d}

\newcommand{\FdEVS}{$\mathcal{F}_d$-{\sc EVS}\xspace}

\newcommand{\black}{{\sf black}\xspace}
\newcommand{\white}{{\sf white}\xspace}
\newcommand{\cutir}{\textsc{Colored Unit 2-interval Recognition}\xspace}
\newcommand{\uivs}{\textsc{Unit-Interval-VS}\xspace}
\newcommand{\uisvs}{\textsc{Unit-Interval-SVS}\xspace}

\usepackage{nameref}
\definecolor{ForestGreen}{rgb}{0.1333,0.5451,0.1333}
\definecolor{DarkRed}{rgb}{0.8,0,0}
\definecolor{Red}{rgb}{1,0,0}
\colorlet{mix}{red!50!black}

\usepackage{hyperref}
\hypersetup{
	pdfnewwindow=true,      % links in new window
	colorlinks=true,       % false: boxed links; true: colored links
	linkcolor=mix,          % color of internal links
	citecolor=ForestGreen,        % color of links to bibliography
	urlcolor=blue          % color of external links
}

\usepackage{tabularx,lipsum,environ}

\makeatletter
\newcommand{\problemtitle}[1]{\gdef\@problemtitle{#1}}% Store problem title
\newcommand{\probleminput}[1]{\gdef\@probleminput{#1}}% Store problem input
\newcommand{\problemquestion}[1]{\gdef\@problemquestion{#1}}% Store problem question
\NewEnviron{problem}{
	\problemtitle{}\probleminput{}\problemquestion{}% Default input is empty
	\BODY% Parse input
	\par\addvspace{.5\baselineskip}
	\noindent
	\begin{tabularx}{\textwidth}{@{\hspace{\parindent}} l X c}
		\multicolumn{2}{@{\hspace{\parindent}}l}{\@problemtitle} \\% Title
		\textbf{Input:} & \@probleminput \\% Input
		\textbf{Question:} & \@problemquestion% Question
	\end{tabularx}
	\par\addvspace{.5\baselineskip}
}

\title{Hardness of Vertex Splitting: Cographs, Chordal Graphs, and Beyond}

\titlerunning{Hardness of Vertex Splitting: Cographs, Chordal Graphs, and Beyond}

 \author{Satyabrata Jana}{Indian Institute of Science Education and Research Berhampur, Odisha, India,  \and \url{https://sites.google.com/view/sbjana}}{satyamtma@gmail.com}{https://orcid.org/0000-0002-7046-0091}{}

 \author{Shivesh K. Roy}{The Institute of Mathematical Sciences, HBNI, Chennai, India \and \url{https://sites.google.com/view/shiveshroy}}{shiveshkr@imsc .res.in}{https://orcid.org/0000-0003-0896-3437}{}

 \author{R. B. Sandeep}{Indian Institute of Technology Dharwad, India \and \url{https://sites.google.com/site/homepagesandeeprb}}{sandeeprb@iitdh.ac.in}{https://orcid.org/0000-0003-4383-1819}{}

\authorrunning{S. Jana, S.K. Roy, and R. B. Sandeep}



\category{}%optional, e.g. invited paper

\relatedversion{}%optional, e.g. full version hosted on arXiv, HAL, or other respository/website
\supplement{}%optional, e.g. related research data, source code, ... hosted on a repository like zenodo, figshare, GitHub, ...

\funding{} 

\nolinenumbers %uncomment to disable line numbering

\newcommand{\defproblem}[3]{
  \vspace{1mm}
\noindent\fbox{
  \begin{minipage}{0.96\textwidth}
  \begin{tabular*}{\textwidth}{@{\extracolsep{\fill}}lr} #1  &  \\ \end{tabular*}
  {\bf{Input:}} #2  \\
  {\bf{Question:}} #3
  \end{minipage}
  }
  \vspace{1mm}
}

\newcommand{\VC}{{\sc Vertex Cover}\xspace}
\newcommand{\PFVS}{{\sc Cograph-VS}\xspace}
\newcommand{\PFEVS}{{\sc Cograph-EVS}\xspace}
\newcommand{\cost}{\operatorname{cost}}
\newcommand{\SEPVC}{{\sc SEPVC}\xspace}
\newcommand{\girth}{\operatorname{girth}}

\usepackage{tcolorbox}
\tcbuselibrary{skins}
\usepackage{xcolor}
\usepackage{mdframed}

\usepackage{booktabs}

\usepackage{nameref}
\definecolor{ForestGreen}{rgb}{0.1333,0.5451,0.1333}
\definecolor{DarkRed}{rgb}{0.8,0,0}
\definecolor{Red}{rgb}{1,0,0}

\usepackage{thm-restate}

\usepackage{thmtools}
\usepackage{cleveref}

\def\DEBUG{true}
\ifdefined\DEBUG{}
\def\rem#1{{\marginpar{\raggedright\scriptsize #1}}}

\newcommand{\skr}[1]{\rem{\textcolor{teal}{skr: #1}}}

\newcommand{\rbr}[1]{\rem{\textcolor{violet}{$\bullet $ #1}}}

\newcommand{\sjr}[1]{\rem{\textcolor{blue}{$\bullet $ #1}}}

\else

\newcommand{\skr}[1]{}

\newcommand{\rbr}[1]{}

\newcommand{\sjr}[1]{}

\fi

\newcommand{\ETH} {{\sf ETH}\xspace}

\newcommand{\fpt} {{\sf FPT}\xspace}

\newcommand{\npc} {{\sf NP}-complete\xspace}

\usepackage{mathtools}

\newcommand{\chovs}{{\sc Chordal-VS}\xspace}
\newcommand{\chovsall}{{\sc Chordal-VS, Chordal-EVS, Chordal-SVS,} and {\sc Chordal-SEVS}\xspace}

\newcommand{\yes}{{\sf YES}\xspace}

\begin{document}

\maketitle

\begin{abstract}
Vertex splitting is a graph modification operation that replaces a vertex~\(v\)
by two nonadjacent vertices whose neighborhoods together equal the
neighborhood of~\(v\). A split is called \emph{exclusive} if the two
neighborhoods are disjoint, and \emph{shallow} if no vertex created by a split
is split again. For a graph property~\(\Pi\), the
\textsc{\(\Pi\)-Vertex Splitting} problem asks, given a graph~\(G\) and a
nonnegative integer~\(k\), whether \(G\) can be transformed into a graph
satisfying~\(\Pi\) using at most~\(k\) vertex splits. Vertex splitting has
recently attracted renewed attention in algorithmic graph theory and
computational geometry~[Firbas and Sorge, ISAAC~2024; Nöllenburg
\emph{et al.}, JoCG~2025], where it has been studied both as a graph
modification operation and as a modeling tool for geometric graph
representations. Motivated by these developments, we continue the systematic
study of vertex splitting problems. Our main contributions are as follows.

\begin{itemize}
    \item We resolve an open problem posed by Firbas and Sorge~[ISAAC~2024]
    by proving that \textsc{Cograph Vertex Splitting} is
    \textsf{NP}-complete, even on graphs of girth at least~$5$.

    \smallskip

    \item Extending %this
    the above result, we prove that
    \textsc{\(P_t\)-free Vertex Splitting} is \textsf{NP}-complete for every
    fixed integer~$t\geq 4$, where a graph is $P_t$-free if it does not contain
    an induced path on~$t$ vertices.

    \smallskip

    \item We further resolve two open problems posed by Abu{-}Khzam,
    Chakraborty, Isenmann, and Oijid~[IWOCA~2026] by showing that both
    \textsc{Chordal Vertex Splitting} and
    \textsc{Unit-Interval Vertex Splitting} are \textsf{NP}-complete.

    \smallskip

    \item Our hardness results extend to the exclusive and shallow variants
    of vertex splitting. Assuming the Exponential Time Hypothesis, none of
    these problems admits an algorithm running in time
    \(2^{o(k)}\cdot n^{O(1)}\). Moreover, except for
    \textsc{Unit-Interval Vertex Splitting} and its variants, none of these problems admits an
    algorithm running in time \(2^{o(n)}\).

\end{itemize}

\end{abstract}

\section{Introduction}
\label{sec:intro}

A \emph{vertex-splitting} operation replaces a vertex $v$ of a graph by two
non-adjacent vertices whose combined neighborhoods equal the neighborhood of $v$.
This operation has appeared in the literature in several contexts.
The \emph{splitting number} of a graph $G$ is the minimum number of vertex-splitting
operations required to transform $G$ into a planar graph.
The study of splitting numbers dates back to the 1980s
(see, e.g.,~\cite{hartsfield1985splitting,jackson1984splitting,jackson1985splittings}),
where the primary focus was on establishing bounds on this parameter.
Faria, de Figueiredo, and de Mendon{\c{c}}a N~\cite{faria2001splitting}
proved that deciding whether the splitting number of a graph is at most $k$
is {\sf NP}-complete even for cubic graphs.
More recently, N{\"o}llenburg et al.~\cite{nollenburg2025planarizing}
obtained a non-uniform fixed-parameter tractable (\fpt) algorithm
for the problem. In fact, they showed a more general result:
the problem admits a non-uniform \fpt algorithm whenever the target graph class is minor-closed.
Vertex splitting has also found applications in improving
the readability of social networks~\cite{y2008improving},
in graph layout problems~\cite{eades1995vertex},
in the construction of chromatic index critical graphs~\cite{hilton1997vertex},
and in kernelization rules for graph modification problems~\cite{sandeep2015parameterized}.

%\medskip
%\noindent
%\textbf{$\Pi$-Vertex Splitting.}
Let $\Pi$ be a graph class.
In the $\Pi$-\textsc{Vertex Splitting} problem, given a graph $G$ and a non-negative integer $k$, the task is to determine whether $G$ can be transformed
into a graph in $\Pi$ using at most $k$ vertex-splitting operations.
Recent years have witnessed a surge of interest in this problem for various
target classes $\Pi$.
Baumann, Pfretzschner, and Rutter~\cite{baumann2023parameterized}
(see also~\cite{DBLP:journals/tcs/BaumannPR24})
obtained a linear kernel when $\Pi$ is the class of graphs of pathwidth at most~1 (for exclusive vertex-splitting). 
They further showed that the problem is \fpt
parameterized by $k$ plus the treewidth of the input graph whenever
$\Pi$ is expressible in MSO$_2$.

Firbas and Sorge~\cite{DBLP:conf/isaac/FirbasS24}
studied the problem extensively for hereditary graph classes.
They obtained polynomial-time algorithms when $\Pi$ is the class of forests,
split graphs, or threshold graphs.
They proved {\sf NP}-completeness when $\Pi$ is the class of cluster graphs,
bipartite graphs, perfect graphs, and several classes defined by forbidden
induced subgraphs with additional connectivity constraints.
Moreover, they established a complete {\sf P}/{\sf NP}-complete dichotomy for hereditary
graph classes characterized by finitely many forbidden induced subgraphs,
each of size at most three.
They also showed para-{\sf NP}-hardness when $\Pi$ is the class of triangle-free graphs
and obtained a linear kernel for the cluster target class parameterized by $k$.
For a comprehensive account of these results, see ~\cite{firbas2023establishing}.

Gaikwad et al.~\cite{gaikwad2025inclusive} showed that the problem is
\npc\ when $\Pi$ is either the class of disjoint unions of stars or the
class of bipartite graphs, but is polynomial-time solvable when $\Pi$ is
the class of disjoint unions of cycles. See also the PhD thesis of
Kumar~\cite{kumar2026parameterized}.
Abu-Khzam and Thoumi~\cite{abu2025complexity}
proved that the problem is {\sf NP}-complete for the class of claw-free graphs
for the \emph{exclusive} variant of vertex splitting.
This resolved an open question posed by Firbas and Sorge~\cite{DBLP:conf/isaac/FirbasS24}.
They also proved {\sf NP}-hardness when forbidding larger star graphs.
More recently, Abu-Khzam, Chakraborty, Isenmann, and Oijid~\cite{abu2026complexity}
proved {\sf NP}-completeness for interval graphs
and showed that the problem is polynomial-time solvable
when the target class consists of disjoint unions of paths.

Several other variants have also been investigated \cite{abu2019cluster,DBLP:conf/ai4i/Abu-KhzamBFS21,DBLP:conf/cocoon/AbuKhzamDIT25,DBLP:journals/corr/abs-1901-00156,DBLP:conf/iwoca/AbuKhzamIM25,DBLP:conf/gd/AhmedKK22}. For example,
Ahmed, Kobourov, and Kryven~\cite{DBLP:conf/gd/AhmedKK22}
obtained an \fpt algorithm for transforming a bipartite graph into one
admitting a two-layer planar drawing via vertex splitting. {\sc Bicluster Editing} with One-sided Vertex Splitting has polynomial kernel \cite{DBLP:conf/iwoca/AbuKhzamIM25}. 
Abu-Khzam et al.~\cite{abu2019cluster}
studied a variant where both vertex splitting and edge editing are allowed,
showing {\sf NP}-completeness and \fpt parameterized by $k$
when the target class is cluster graphs.

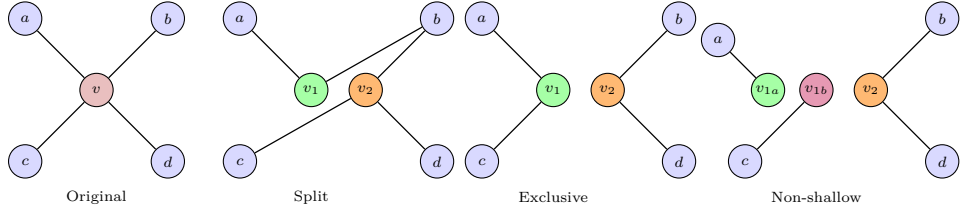
\begin{figure}[t]
\centering
\scalebox{0.80}{
\begin{tikzpicture}[
    every node/.style={circle, draw=black, line width=0.5pt,
        minimum size=5.5mm, inner sep=0pt, font=\scriptsize},
    main/.style={fill=red!65!black!25},
    neigh/.style={fill=blue!15},
    splitA/.style={fill=green!35},
    splitB/.style={fill=orange!55},
    splitC/.style={fill=purple!40},
    edge/.style={line width=0.6pt}
]

%%%%%%%%%%%%%%%%%%%%%%%%%%%%
%% (1) Original
%%%%%%%%%%%%%%%%%%%%%%%%%%%%
\begin{scope}[xshift=0cm]

\node[main] (v) {$v$};

\node[neigh, above left=1.1cm of v] (a) {$a$};
\node[neigh, above right=1.1cm of v] (b) {$b$};
\node[neigh, below left=1.1cm of v] (c) {$c$};
\node[neigh, below right=1.1cm of v] (d) {$d$};

\draw[edge] (v)--(a);
\draw[edge] (v)--(b);
\draw[edge] (v)--(c);
\draw[edge] (v)--(d);

\node[draw=none, rectangle, below=1.2cm of v] {\scriptsize Original};

\end{scope}

%%%%%%%%%%%%%%%%%%%%%%%%%%%%
%% (2) General split
%%%%%%%%%%%%%%%%%%%%%%%%%%%%
\begin{scope}[xshift=4.0cm]

\node[splitA] (v1) at (-0.45,0) {$v_1$};
\node[splitB] (v2) at (0.45,0) {$v_2$};

\node[neigh, above left=1.1cm of v1] (a2) {$a$};
\node[neigh, above right=1.1cm of v2] (b2) {$b$};
\node[neigh, below left=1.1cm of v1] (c2) {$c$};
\node[neigh, below right=1.1cm of v2] (d2) {$d$};

\draw[edge] (v1)--(a2);
\draw[edge] (v1)--(b2);
\draw[edge] (v2)--(b2);
\draw[edge] (v2)--(c2);
\draw[edge] (v2)--(d2);

\node[draw=none, rectangle, below=1.2cm of v1] {\scriptsize Split};

\end{scope}

%%%%%%%%%%%%%%%%%%%%%%%%%%%%
%% (3) Exclusive shallow
%%%%%%%%%%%%%%%%%%%%%%%%%%%%
\begin{scope}[xshift=8.0cm]

\node[splitA] (u1) at (-0.45,0) {$v_1$};
\node[splitB] (u2) at (0.45,0) {$v_2$};

\node[neigh, above left=1.1cm of u1] (a3) {$a$};
\node[neigh, below left=1.1cm of u1] (c3) {$c$};

\node[neigh, above right=1.1cm of u2] (b3) {$b$};
\node[neigh, below right=1.1cm of u2] (d3) {$d$};

\draw[edge] (u1)--(a3);
\draw[edge] (u1)--(c3);
\draw[edge] (u2)--(b3);
\draw[edge] (u2)--(d3);

\node[draw=none, rectangle, below=1.2cm of u1] {\scriptsize Exclusive};

\end{scope}

%%%%%%%%%%%%%%%%%%%%%%%%%%%%
%% (4) Non-shallow split
%%%%%%%%%%%%%%%%%%%%%%%%%%%%
\begin{scope}[xshift=12.0cm]

% v1 split again
\node[splitA] (x1) at (-0.9,0) {$v_{1a}$};
\node[splitC] (x2) at (-0.1,0) {$v_{1b}$};
\node[splitB] (x3) at (0.8,0) {$v_2$};

\node[neigh, above left=.6cm of x1] (a4) {$a$};
\node[neigh, below left=1.1cm of x2] (c4) {$c$};

\node[neigh, above right=1.1cm of x3] (b4) {$b$};
\node[neigh, below right=1.1cm of x3] (d4) {$d$};

\draw[edge] (x1)--(a4);
\draw[edge] (x2)--(c4);

\draw[edge] (x3)--(b4);
\draw[edge] (x3)--(d4);

\node[draw=none, rectangle, below=1.2cm of x2] {\scriptsize Non-shallow};

\end{scope}

\end{tikzpicture}
}
\caption{\small Vertex splitting: original, general split, exclusive split, and non-shallow split.}
\end{figure}\label{fig:split}

\medskip
\noindent

We study vertex splitting and three natural restrictions of it.
A split is \emph{exclusive} if the neighborhoods of the two new vertices are disjoint.
%\srb{we have defined exclusive many times in the introduction}
%\sbj{Fixed. removed previous occurrences. As Here all the variant is defined.}
The problem is \emph{shallow} if no vertex created by a split is split again.
See Figure \ref{fig:split} for an illustration.
Accordingly, we use the suffix
\textsc{VS} for the general problem,
\textsc{EVS} for the exclusive variant,
\textsc{SVS} for the shallow variant, and
\textsc{SEVS} for the shallow exclusive variant. To avoid the repetitive definitions, we only define VS variants of each problem; the other variants, EVS, SVS, and SEVS are defined analogously. For a graph class~$\Pi$, the {\sc $\Pi$-VS} problem is formally defined as follows.

\defproblem{{\sc $\Pi$-Vertex Splitting} ($\Pi$-VS)}{A graph $G$, and a non-negative integer $k$.}
{Can $G$ be transformed into a graph belonging to the class~$\Pi$ using at most $k$ vertex splits?}

Our first main result resolves an open problem posed by Firbas and Sorge~\cite{DBLP:conf/isaac/FirbasS24} by proving that
\textsc{Cograph-Vertex Splitting} is \textsf{NP}-complete, even on graphs of girth at least~$5$.

\begin{restatable}{theorem}{theoremp4}
\label{thm:PFEVS-PFVS-NPcomplete}
For every fixed integer $g\geq 5$, 
\PFVS and \PFEVS are
\npc on subcubic graphs of girth at least $g$.
Furthermore, assuming {\sf ETH},
neither  can be  solved in time 
$2^{o(n+m)}$ nor in time $2^{o(k)}\cdot n^{O(1)}$.
\end{restatable}

Since cographs are precisely the graphs that do not contain~$P_4$ (path graph on 4 vertices) as an induced subgraph, a natural next step is to study the problem of \textsc{$P_t$-free-VS} for every fixed~$t \geq 5$. Towards this, we first show that \PfiveFVS is \npc on graphs of girth at least 6.

\begin{restatable}{theorem}{theoremp5}
\label{thm:P5VS-NPcomplete}
For every fixed integer $g\geq 6$, \PfiveFVS and \PfiveFEVS
are \npc on graphs $G$ of girth at least $g$ satisfying $|E(G)|\geq |V(G)|$. 
Furthermore, assuming {\sf ETH}, neither  can be  solved in time $2^{o(n+m)}$ nor in time $2^{o(k)}\cdot n^{O(1)}$. 
\end{restatable}

Then we prove the following theorem about  {\sc $P_t$-free-VS} for all fixed~$t \ge 6$ to complete the picture about~$P_t$-free graphs.

\begin{restatable}{theorem}{theorempt}
\label{thm:pt}
For every fixed integer $t\geq 6$, {\sc $P_t$-free-VS},
{\sc $P_t$-free-EVS}, {\sc $P_t$-free-SVS}, and
{\sc $P_t$-free-SEVS} are {\sf NP}-complete. Moreover, assuming {\sf ETH},
none of them  can be solved in time  $2^{o(n)}$ nor in time $2^{o(k)}\cdot n^{O(1)}$. 
\end{restatable}

Thus, \cref{thm:PFEVS-PFVS-NPcomplete}, \cref{thm:P5VS-NPcomplete}, and \cref{thm:pt} together imply the following general result: \textsc{$P_t$-free-Vertex Splitting} is \textsf{NP}-complete for every fixed~$t \geq 4$.

Our next main results resolve two open questions posed by Abu-Khzam et al.~\cite{abu2026complexity} by showing {\sf NP}-hardness for \textsc{Chordal Vertex Splitting} and its variants by a reduction from \textsc{Chain Vertex Deletion},
and  \textsc{Unit-Interval Vertex Splitting} and its shallow variant by a reduction from \cutir. 

\begin{restatable}{theorem}{theoremchordal}\label{thm:chordal}
\chovsall are \npc. Moreover, assuming {\sf ETH}, none of them can be solved in time $2^{o(n)}$ nor in time $2^{o(k)}\cdot n^{O(1)}$.
\end{restatable}

\begin{restatable}{theorem}{theoremui}\label{thm:ui}
{\sc \uivs} and {\sc \uisvs} are \npc. Moreover, assuming {\sf ETH}, neither problem can be solved in time
\(2^{o(\sqrt n)}\), nor in time \(2^{o(k)}\cdot n^{O(1)}\).
\end{restatable}

\subparagraph{Organization.}
After introducing the necessary preliminaries in \Cref{sec:prelims}, we establish several auxiliary
results in \Cref{sec:equi}. We then prove the hardness of {\sc Cograph-VS}
and its variants in \Cref{sec:cograph}, followed by the hardness of
{\sc \(P_5\)-free-VS} and its variants in \Cref{sec:p5}. Building on these
results, in \Cref{sec:pt} we extend the hardness to {\sc \(P_t\)-free-VS}
and its variants for every fixed \(t \ge 6\).
Next, in \Cref{sec:chordal}, we obtain hardness results for \textsc{Chordal-VS} and its variants. Finally, in \Cref{sec:unit-interval} we obtain hardness for \uivs and its shallow variant.
%Finally, in
%\Cref{sec:bipnew}, we show that {\sc Bipartite-VS} and its variants are
%equivalent to {\sc Odd Cycle Transversal}.
\ifthenelse{\boolean{shortver}}{{A full version of the paper containing all missing proofs  is appended at the end.}
}{}

\section{Preliminaries}
\label{sec:prelims}

In this section, we review some basic preliminaries. For notions not defined here, we refer the reader to the textbooks~\cite{CyganFKLMPPS15, Diestel17}. Throughout the paper, we consider simple unweighted undirected graphs. For a graph~$G$, we use~$V(G)$ and~$E(G)$ to denote the vertex and edge sets of~$G$, respectively. We use the shorthand~$n \coloneqq |V(G)|$ and~$m \coloneqq |E(G)|$. 
%\srb{we never used antiedges, right?}
%The set of anti-edges is denoted by $E(\bar{G})$, i.e., \sk{$E(\bar{G}) = \{uv \mid u \neq v \text{ and } uv \notin E(G)\}$}. 
For a positive integer~$\ell$ we use~$[\ell]$ to denote the set~$\{1, 2, \dots, \ell\}$.
The open neighborhood of a vertex $v \in V(G)$ is defined as $N_G(v) := \{ u \mid \{u,v\} \in E(G) \}$, where we omit the subscript $G$ whenever the underlying graph is clear from the context. For a vertex set $S \subseteq V(G)$, the open neighborhood of $S$, denoted $N_G(S)$, is defined by $N_G(S) := \left( \bigcup_{v \in S} N_G(v) \right) \setminus S$, that is, the set of all vertices adjacent to at least one vertex in $S$, excluding the vertices of $S$ itself. 
%\srb{did we ever use $N(S)$?\sk{I am not quite sure, I think keeping it will not harm at least ;)}}
%\srb{todo: remove unwanted definitions.}
For $S \subseteq V(G)$, the subgraph induced by $S$ is denoted by $G[S]$.
For $F\subseteq E(G)$, the graph $(V(F), F)$ is denoted by $G[F]$, where $V(F)$ is the set of vertices incident with edges in $F$.
For two vertices $x,y \in V(G)$, an $x$-$y$ path in $G$ is a sequence of vertices $(x = v_1, \ldots, v_k = y)$ such that $\{v_i, v_{i+1}\} \in E(G)$ for all $i \in [k-1]$.  The {\em girth} of a graph, denoted by $\girth(G)$, is the length of its shortest cycle.
The $\girth$ of a forest is infinity.
%
%\sbj{keep all the graph class definitions in one para.}\skr{Done}
%
A graph is {\em $P_t$-free} if it has no induced path on $t$ vertices. A graph is a {\em cograph} if it is $P_4$-free. A {\em house} graph is the complement of a \(P_{5}\). For every fixed integer \(d\geq 0\), we use $\mathcal{F}_d$ to denote the class of forests in which each component has diameter at most $d$. %A  {\em star}   is a connected graph in which at most one vertex has degree greater than one. 
A {\em star} graph is a tree with at least two vertices and consists of one central vertex adjacent to all other vertices, which are leaves. A  {\em double-star}   is a graph that is formed by two stars via joining their centers by an edge.

\subparagraph{Vertex Splitting.}
For an integer $i \ge 0$, let $G_i$ be a graph and let $v_i \in V(G_i)$.
Let $V_i^a$ and $V_i^b$ be (possibly empty) subsets of $N_{G_i}(v_i)$ such that~$V_i^a \cup V_i^b = N_{G_i}(v_i)$. The graph $G_{i+1}$ is obtained from $G_i$ by \emph{splitting} $v_i$ according to the pair $(V_i^a, V_i^b)$, that is, by deleting $v_i$, introducing two new non-adjacent vertices $v_i^a$ and $v_i^b$, making $v_i^a$ adjacent to every vertex in $V_i^a$, and making $v_i^b$ adjacent to every vertex in $V_i^b$. All other adjacencies of $G_i$ remain unchanged.
A sequence of $k$ vertex splits on a graph $G$ is defined by setting
$G_0 = G$ and, for each $0 \le i \le k-1$, selecting a vertex
$v_i \in V(G_i)$ together with subsets
$V_i^a, V_i^b \subseteq N_{G_i}(v_i)$ satisfying
$V_i^a \cup V_i^b = N_{G_i}(v_i)$.
The graph $G_k$ denotes the graph obtained after performing
these $k$ vertex splits.

The \emph{descendants} of a vertex $v$ are the vertices obtained from $v$ by splitting $v$ and subsequently splitting any vertices derived from it.
By \emph{splitting a vertex $v$ multiple times}, we mean splitting $v$ once and then recursively splitting its descendants.
When the order of the splits is clear from the context,
we may omit explicit reference to the intermediate graphs $G_i$. 

\begin{observation}
\label{obs:kbound}
Let \(G\) be a graph, let \(k\) be a
nonnegative integer, and let \(\Pi\) be a graph property. Then the following
hold.
\begin{enumerate}
    \item Suppose that every graph whose components are isomorphic to
    \(K_1\) or \(K_2\) satisfies \(\Pi\). If \(k\geq 2m\), then
    \((G,k)\) is a \yes-instance of both
    \textsc{\(\Pi\)-VS} and \textsc{\(\Pi\)-EVS}.

    \item If \(k\geq n\), then \((G,k)\) is a \yes-instance of
    \textsc{\(\Pi\)-SVS} if and only if \((G,n)\) is. The analogous
    statement holds for \textsc{\(\Pi\)-SEVS}.
\end{enumerate}
\end{observation}

\begin{proof}
For the first statement, split every nonisolated vertex \(v\) exclusively
into \(d_G(v)\) copies, each incident with exactly one original edge. This
requires \(d_G(v)-1\) splits. Hence the total number of splits is
\[
    \sum_{v:d_G(v)>0}\bigl(d_G(v)-1\bigr)\leq
    \sum_{v\in V(G)}d_G(v)=2m.
\]
The resulting graph is a disjoint union of \(K_1\)'s and \(K_2\)'s and
therefore satisfies \(\Pi\). Since every exclusive split is also a general
split, the conclusion holds for both variants.

For the second statement, every shallow split sequence contains at most one
split for each original vertex and hence at most \(n\) splits. Thus, for
\(k\geq n\), the budgets \(k\) and \(n\) yield the same yes-instances. The
same argument applies to shallow exclusive splitting.
\end{proof}

All target graph classes \(\Pi\) considered in this paper---namely,
\(P_t\)-free graphs for fixed \(t\geq 4\), chordal graphs, and unit interval
graphs---can be recognized in polynomial time, and each contains every graph
whose components are isomorphic to \(K_1\) or \(K_2\). By
\Cref{obs:kbound}, for \textsc{\(\Pi\)-VS} and
\textsc{\(\Pi\)-EVS}, the budget may be restricted to \(k<2|E(G)|\), while
for \textsc{\(\Pi\)-SVS} and \textsc{\(\Pi\)-SEVS}, it may be capped at
\(|V(G)|\). Thus every \yes-instance admits a polynomial-size split sequence
that can be verified in polynomial time. Consequently, all vertex-splitting
problems considered in this paper belong to \(\mathsf{NP}\), and we omit
separate proofs of membership in \(\mathsf{NP}\) when establishing
\(\mathsf{NP}\)-hardness.

In computational complexity theory, the Exponential Time Hypothesis ({\sf ETH})
is a computational hardness assumption stating that 3-SAT, the satisfiability
problem for 3-CNF Boolean formulas with $n$ variables, cannot be solved in time
$2^{o(n)}$~\cite{DBLP:journals/jcss/ImpagliazzoP01}.
For a detailed exposition of lower bounds based \ETH, we refer to the textbook~\cite{CyganFKLMPPS15}.
By $\leq_{\mathtt{poly}}$ we denote polynomial-time many to one reductions.

One of the source problems that we use for our reductions is \VC: given a graph $G$ and a nonnegative integer $k$, the objective
is to decide whether $G$ has a vertex cover of size at most $k$.
We next record a complexity result for \VC on sparse graphs of large
girth. Although Komusiewicz \cite{komusiewicz2018tight} does not explicitly impose the additional
restriction that every connected component of the input graph contains
a cycle, this restriction is without loss of generality: minimum vertex
covers of tree components can be computed in polynomial time, after
which those components can be deleted and the budget adjusted
accordingly.

\begin{proposition}{\rm \cite[Theorem~2.1]{komusiewicz2018tight}}
\label{prop:vc-high-girth}
For every fixed integer \(g \geq 3\), \VC is {\sf NP}-complete on
subcubic graphs of girth at least \(g\) in which every connected
component contains a cycle. Moreover, assuming {\sf ETH}, \VC cannot
be solved on this graph class in time \(2^{o(n+m)}\).
\end{proposition}

\section{Polynomial equivalence for Vertex Splitting on high-girth graphs}
\label{sec:equi}

In this section, we establish polynomial equivalences between several
vertex-splitting variants on graphs of sufficiently large girth. In particular,
we relate \(P_t\)-free targets to forests with bounded-diameter components and show that, in
this high-girth setting, arbitrary splits and exclusive splits are equivalent
for \(P_t\)-free target graphs. We begin with an observation on the effect of exclusive vertex splitting on girth of a graph.

\begin{observation}%\ifthenelse{\boolean{shortver}}{$(\bigstar)$}{}
\label{obs:no-new-cycle}
Exclusive vertex splitting does not decrease the girth.
\end{observation}
%\srb{This observation is used by both $P_4$-free and $P_5$-free. Hence brought here}

\ifthenelse{\boolean{shortver}}{}{
\begin{proof}
Let $G'$ be obtained from $G$ by splitting a vertex $v$ into two nonadjacent
vertices $v_1$ and $v_2$. Suppose, for contradiction, that $G'$ contains a cycle
$C$ of length $t$ such that $t < \girth(G)$. Since $G$ itself has no cycle of length smaller than $\girth(G)$, the cycle $C$
must use both $v_1$ and $v_2$; otherwise $C$ would already be a cycle in $G$. %Let $P$ be a subpath of $C$ between $v_1$ and $v_2$. 
Because the split is exclusive, $v_1$ and $v_2$ do not have a common neighbor.
By identifying $v_1$ and $v_2$ back into the original vertex $v$,
the cycle $C$ becomes a closed walk in $G$ of length at most $t$ and at least 3.
Every closed walk contains a simple cycle, hence $G$ contains
a cycle of length at most $t$, contradicting $\girth(G)>t$. 
\end{proof}
}

The next lemma relates forests of bounded component diameter to \(P_{d+2}\)-free graphs.

\begin{lemma}
\label{lem:forest-diameter-equivalent-to-Pdplus2-free-high-girth}
Let \(d\geq 0\) be an integer, and let \(H\) be a graph of girth at least
\(d+3\). Then \(H\) is a forest in which every connected component has
diameter at most \(d\) if and only if \(H\) is \(P_{d+2}\)-free.
\end{lemma}

\ifthenelse{\boolean{shortver}}{}{
\begin{proof}
First suppose that \(H\) is a forest in which every connected component has
diameter at most \(d\). We show that \(H\) is \(P_{d+2}\)-free. Suppose, for a
contradiction, that \(H\) contains an induced copy of \(P_{d+2}\). Then some
connected component of \(H\) contains a path on \(d+2\) vertices, that is, a
path of length \(d+1\). Hence that connected component has diameter at least
\(d+1\), contradicting the assumption that every connected component has
diameter at most \(d\). Therefore \(H\) is \(P_{d+2}\)-free.

Conversely, suppose that \(H\) is \(P_{d+2}\)-free. We first prove that \(H\)
is acyclic. Suppose, for a contradiction, that \(H\) contains a cycle. Let
\(C\) be a shortest cycle in \(H\). Since \(H\) has girth at least \(d+3\),
the length of \(C\) is at least \(d+3\). Moreover, \(C\) is chordless;
otherwise a chord of \(C\) would give a shorter cycle. Hence any \(d+2\)
consecutive vertices of \(C\) induce a copy of \(P_{d+2}\), contradicting the
assumption that \(H\) is \(P_{d+2}\)-free. Thus \(H\) is a forest.

It remains to show that every connected component of \(H\) has diameter at
most \(d\). Suppose, for a contradiction, that some connected component of
\(H\) has diameter at least \(d+1\). Then this component contains a path on
\(d+2\) vertices. Since \(H\) is a forest, every path in \(H\) is induced.
Therefore \(H\) contains an induced copy of \(P_{d+2}\), contradicting the
assumption that \(H\) is \(P_{d+2}\)-free. Hence every connected component of
\(H\) has diameter at most \(d\).

Thus \(H\) is a forest in which every connected component has diameter at
most \(d\), completing the proof.
\end{proof}
}

We next show that, on high-girth graphs, arbitrary splits and exclusive splits are equivalent for \(P_t\)-free targets.

\begin{lemma}\label{lem:Pt-free-VS-equivalent-EVS-high-girth}
Let \(t\geq 4\) be a fixed integer, and let \(g\geq t+1\) be an integer.
Let \(G\) be a graph of girth at least \(g\). Then, for every integer
\(k\geq 0\), \((G,k)\) is a \yes-instance of
{\sc \(P_t\)-free-VS} if and only if \((G,k)\) is a \yes-instance of
{\sc \(P_t\)-free-EVS}.
\end{lemma}

%\ifthenelse{\boolean{shortver}}{}{
\begin{proof}
Since every exclusive split is, in particular, a split, one implication is
immediate. Thus, if \((G,k)\) is a \yes-instance of
\textsc{\(P_t\)-free-EVS}, then it is also a \yes-instance of
\textsc{\(P_t\)-free-VS}.

Conversely, suppose that \((G,k)\) is a \yes-instance of
\textsc{\(P_t\)-free-VS}. Let \(H\) be a \(P_t\)-free graph obtained from
\(G\) by at most \(k\) splits. Consider the natural map from \(V(H)\) to
\(V(G)\), which sends every vertex of \(H\) to the original vertex of \(G\)
from which it descends. For every \(v\in V(G)\), let \(X_v\) be the set of
descendants of \(v\) in \(H\). Then \(\{X_v : v\in V(G)\}\) is a partition of
\(V(H)\), and each \(X_v\) is an independent set. Moreover,
\[
    \sum_{v\in V(G)} (|X_v|-1) \leq k .
\]

For every edge \(uv\in E(G)\), there is at least one edge of \(H\) with one
endpoint in \(X_u\) and the other endpoint in \(X_v\). Choose exactly one
such edge and denote it by \(e_{uv}\). Define a spanning subgraph \(H'\) of
\(H\) by
\[
    V(H')=V(H)
    \quad\text{and}\quad
    E(H')=\{e_{uv}: uv\in E(G)\}.
\]
We first show that \(H'\) can be obtained from \(G\) by at most $k$ exclusive splits.
The partition \(\{X_v : v\in V(G)\}\) has the following properties:
\begin{enumerate}
    \item each \(X_v\) is an independent set;
    \item for every edge \(uv\in E(G)\), there is exactly one edge of \(H'\)
    between \(X_u\) and \(X_v\);
    \item for every non-edge \(uv\notin E(G)\), there is no edge of \(H'\)
    between \(X_u\) and \(X_v\).
\end{enumerate}

We claim that any graph satisfying these properties can be obtained from \(G\)
by exclusive splits.

We prove the claim by induction on
$
    r=\sum_{v\in V(G)} (|X_v|-1).
$
If \(r=0\), then every set \(X_v\) is a singleton, and the graph is precisely
\(G\). Assume that \(r>0\). Choose a vertex \(v\in V(G)\) with
\(|X_v|\geq 2\), and choose two distinct vertices \(x,y\in X_v\). 
Identify \(x\) and \(y\) into a single vertex \(z\), whose neighbourhood is  $N(z)=N_{H'}(x)\cup N_{H'}(y)$.
Let the resulting graph be \(\widehat{H}\). In the corresponding partition,
the set \(X_v\) is replaced by  $(X_v\setminus \{x,y\})\cup \{z\}$,
and all other sets \(X_u\) remain unchanged. The three properties above are
preserved, and the value of \(r\) decreases by one. Hence, by the induction
hypothesis, \(\widehat{H}\) can be obtained from \(G\) by exclusive splits.

Now reverse the identification of \(x\) and \(y\). That is, split \(z\) into
\(x\) and \(y\), assigning to \(x\) precisely the neighbours it has in \(H'\)
and assigning to \(y\) precisely the neighbours it has in \(H'\). This split
is exclusive: no vertex is adjacent to both \(x\) and \(y\), because for every
edge \(uv\in E(G)\) there is exactly one edge between \(X_u\) and \(X_v\) in
\(H'\). Therefore the reverse operation is an exclusive split, and it
reconstructs \(H'\).

Thus \(H'\) can be obtained from \(G\) by exactly
$$
    \sum_{v\in V(G)} (|X_v|-1) \leq k
$$
exclusive splits.

It remains to prove that \(H'\) is \(P_t\)-free. By
\Cref{obs:no-new-cycle}, \(H'\) has girth at least \(g\), since \(H'\) is
obtained from \(G\) by exclusive splits.
Suppose, for a contradiction, that \(H'\) contains an induced copy of
\(P_t\). Let  $P=p_1p_2\cdots p_t$
be such an induced path in \(H'\). Since \(H\) is \(P_t\)-free and \(H'\) is
a spanning subgraph of \(H\), the vertices \(p_1,\ldots,p_t\) cannot induce a
copy of \(P_t\) in \(H\). Hence there exist two nonconsecutive vertices
\(p_i,p_j\) of \(P\), with \(1\leq i<j\leq t\) and \(j-i\geq 2\), such that
\(p_ip_j\in E(H)\).

Let \(a_i\) and \(a_j\) be the original vertices of \(G\) from which
\(p_i\) and \(p_j\) descend, respectively. Since \(p_ip_j\in E(H)\), we have
\(a_ia_j\in E(G)\). Now consider the subpath $p_ip_{i+1}\cdots p_j$
of \(P\). Its image under the natural map is a trail in \(G\) from \(a_i\) to
\(a_j\). Indeed, if two distinct edges of this subpath were mapped to the
same edge of \(G\), then \(H'\) would contain two distinct edges between the
same two bags, contradicting the construction of \(H'\).

Together with the edge \(a_ia_j\), this trail contains a cycle in \(G\) of
length at most  $(j-i)+1 \leq t$.
Since \(g\geq t+1\), this contradicts the assumption that \(G\) has girth at
least \(g\). Therefore \(H'\) is \(P_t\)-free.
Thus \(H'\) is a \(P_t\)-free graph obtained from \(G\) by at most \(k\)
exclusive splits. Hence \((G,k)\) is a \yes-instance of
\textsc{\(P_t\)-free-EVS}. This completes the proof.
\end{proof}
%}

\section{Cograph Vertex Splitting}
\label{sec:cograph}

In this section, we prove \Cref{thm:PFEVS-PFVS-NPcomplete}. This answers a question raised by Firbas and  Sorge \cite{DBLP:conf/isaac/FirbasS24}. 
We first formulate an edge-partition problem that captures the effect
of exclusive vertex splitting when transforming a graph into a cograph.
Throughout this section, we assume that all input graphs are assumed to have no isolated vertices.
%\subsection{Star Edge Partition with Vertex Cost (\SEPVC)}

\begin{definition}[Star edge partition]
  {\em   Let $G$ be a graph without any isolated vertices. A collection $\mathcal P=\{H_1,\ldots,H_r\}$ of subgraphs of $G$ is a
\emph{star edge partition} if
\begin{itemize}
\item each $H_i$ is a star,
\item the $H_i$ are edge-disjoint,
\item $\bigcup_i E(H_i)=E(G)$.
\end{itemize}}
\end{definition}

The \emph{cost} of a star edge partition  $\mathcal{P}$ is denoted by $ \cost(\mathcal{P})$ where 
    \[
        \cost(\mathcal{P}) \coloneqq \sum_{v \in V(G)} \cost(v); ~~~~~\text{and}   ~~~~~\cost(v) \coloneqq \bigl|\{\, i \in [r] \mid v \in V(H_i) \,\}\bigr| - 1.
    \]

To capture this decomposition algorithmically, we formulate a partitioning problem
that measures how many times a vertex must be ``separated'' among different stars.
This will exactly correspond to the number of vertex splits required. Now we define a problem called {\sc Star Edge Partition with Vertex Cost} (in short, \SEPVC).  This is essentially the same as the \textsc{Weighted Star Decomposition} problem defined in \cite{gaikwad2025inclusive}, with a minor difference in the cost function.

\begin{problem}
	\problemtitle{{\sc Star Edge Partition with Vertex Cost} (\SEPVC)}
	\probleminput{A graph $G$ without any isolated vertices, and a nonnegative integer $k$.}
	\problemquestion{Does $G$ admit a star edge partition of cost at most $k$?}
\end{problem}
% \srb{some problems are in box, some are like this. Need to make it uniform?}

% \sbj{My plan is to keep the problem that we studied is in box, all other problems we use another environment. }
% \srb{I agree.}

\begin{comment}
\subsection{{\sc Union of Stars-EVS} equals {\sc SEPVC}}

We next show that exclusive vertex splitting is precisely the operation needed to
separate the incident edges of a vertex among different stars.
Consequently, minimizing the number of exclusive vertex splits required to transform
a graph into a disjoint union of stars is equivalent to minimizing the vertex cost of
a star edge partition.

\end{comment}
\begin{lemma}\label{lem:SEPVC-equivalent-to-PFEVS}
Let $G$ be a graph without any isolated vertices and with $\girth(G)\ge 5$, and let $k\in\mathbb N$.
Then $(G,k)$ is a \yes-instance of {\sc SEPVC} if and only if $(G,k)$ is a \yes-instance of \PFEVS.
\end{lemma}

\begin{proof}
\noindent\textbf{(\(\Rightarrow\)).}
Assume that $(G,k)$ is a \yes-instance of \SEPVC.
Let $\mathcal{P}=\{H_1,\dots,H_r\}$ be a star edge partition of $G$ with
$\cost(\mathcal{P})\le k$.
For each vertex $v\in V(G)$, let
\[
p(v)\coloneqq\bigl|\{\,i\in[r]\mid v\in V(H_i)\,\}\bigr|
\]
denote the number of stars containing $v$.
Since the contribution of $v$ to the cost is $p(v)-1$, we have

\[
\sum_{v\in V(G)}(p(v)-1)\le k.
\]

We construct a graph $G'$ from $G$ by performing exclusive vertex splits.
For every vertex $v\in V(G)$, if $p(v)=1$, then $v$ is left unchanged (since there are no isolated vertices $p(v)\neq 0$).
Otherwise, we perform $p(v)-1$ successive exclusive vertex splits so that the edges incident with $v$ belonging to different stars of $\mathcal P$ become incident with different copies of $v$.
Equivalently, in the resulting graph each copy of $v$ is incident precisely with the edges of one star of $\mathcal P$ containing $v$. Hence the total number of exclusive vertex splits is
$
\sum_{v\in V(G)}(p(v)-1)\le k.
$

Now consider a star $H_i\in\mathcal P$.
By construction, all edges of $H_i$ are incident only with the copies of vertices created for $H_i$, and no such copy is incident with an edge belonging to another star.
Therefore the vertices corresponding to $H_i$ induce a connected component of $G'$ isomorphic to $H_i$.
Since this holds for every $H_i\in\mathcal P$, every connected component of $G'$ is a tree with diameter at most 2.
%Then by \Cref{lem:forest-diameter-equivalent-to-Pdplus2-free-high-girth}, 
Therefore, $G'$ is $P_4$-free. 
%Thus $G'$ is a disjoint union of stars obtained using at most $k$ exclusive vertex splits.
Hence $(G,k)$ is a \yes-instance of \PFEVS.

\medskip
\noindent\textbf{(\(\Leftarrow\)).}
Conversely, suppose that $G$ can be transformed into a cograph $G'$ using at most $k$ exclusive vertex splits. Each edge of $G$ gives rise to a unique edge of $G'$; exclusive vertex splits merely redistribute the incidences of an edge among the copies of its endpoints.
Hence there is a natural bijection
$
\phi \colon E(G)\to E(G').
$

By \Cref{obs:no-new-cycle}, $\girth(G')\geq \girth(G)\geq 5$.
Then by \Cref{lem:forest-diameter-equivalent-to-Pdplus2-free-high-girth}, each connected component of $G'$ is a tree with diameter at most $2$.
Let $\mathcal C$ be the set of connected components of $G'$ containing at least one edge.
For each component $C\in\mathcal C$, let
$
F_C=\{\,e\in E(G)\mid \phi(e)\in E(C)\,\}.
$
Then $\{F_C:C\in\mathcal C\}$ is a partition of $E(G)$.
Since every connected component, containing at least one edge, of $G'$ is a star, each subgraph $G[F_C]$ is also a star.
Consequently,
$
\mathcal P=\{\,G[F_C]:C\in\mathcal C\,\}
$
is a star edge partition of $G$.

Now fix a vertex $v\in V(G)$ and let $t(v)\ge1$ denote the number of copies of $v$ in $G'$, and $q(v)$ ($1\leq q(v)\leq t(v)$)
be the number of 
nonisolated copies of $v$ in $G'$.
The edges incident with $v$ are distributed among these $q(v)$ copies, and therefore $v$ belongs to exactly $q(v)$ stars of $\mathcal P$.
Hence the contribution of $v$ to the cost of $\mathcal P$ is
$
q(v)-1.
$
Summing over all vertices,
\[
\cost(\mathcal P)
=\sum_{v\in V(G)}(q(v)-1)\leq \sum_{v\in V(G)}(t(v)-1)
=\text{(\# exclusive splits used to obtain $G'$)}
\le k.
\]
Therefore $(G,k)$ is a \yes-instance of \SEPVC.
\end{proof}

%\subsection{{\sf NP}-harness of {\sc \SEPVC}}

%\ifthenelse{\boolean{shortver}}{}{
\subsection{Hardness of {\sc Star Edge Partition with Vertex Cost}}
%}
We prove that {\sc SEPVC} is {\sf NP}-hard by giving a polynomial-time reduction from {\sc Vertex Cover}. 
The reduction is essentially the same reduction used for
\textsc{Weighted Star Decomposition} by Gaikwad et al.\cite{gaikwad2025inclusive}; using the
high-girth hardness of \textsc{Vertex Cover} allows us to retain an
arbitrarily large girth.
The reduction relies on two simple observations. First, for a fixed input graph, the cost of a star edge partition depends only on the number of stars in the partition. Second, the minimum number of stars in a star edge partition equals the size of a minimum vertex cover.

\begin{lemma}\label{lem:VC-to-SEPVC}
Let $G$ be a graph such that every component contains a cycle, and $k$ be a nonnegative integer.
%Let $(G,k)$ be an instance of {\sc \VC} with $n$ vertices and $m$ edges.
Then $(G,k)$ is a \yes-instance of {\sc \VC}\ if and only if $(G, k+m-n)$ is a \yes-instance of {\sc \SEPVC}.
\end{lemma}

\ifthenelse{\boolean{shortver}}{}{
\begin{proof}
Since $G$ has no component isomorphic to a tree, we have $k+m-n\geq 0$.
Fix any star-partition $\mathcal{P}$ of $E(G)$ and let $r\coloneqq |\mathcal{P}|$ be the number of stars.
For a vertex $v\in V(G)$, let $p(v)$ denote the number of stars in $\mathcal{P}$ that contain $v$.
By definition,
\[
\cost(\mathcal{P})=\sum_{v\in V(G)} (p(v)-1)=\Bigl(\sum_{v\in V(G)} p(v)\Bigr)-n.
\]
On the other hand, $\sum_{v\in V(G)} p(v)=\sum_{H\in\mathcal{P}} |V(H)|$ counts the total number of
vertex-incidences over all parts.
If a part $H$ is a star with $|E(H)|=\ell$, then $|V(H)|=\ell+1$.
Since the parts form an edge partition of $E(G)$, we have $\sum_{H\in\mathcal{P}} |E(H)|=m$, and hence
\[
\sum_{H\in\mathcal{P}} |V(H)|
=\sum_{H\in\mathcal{P}} (|E(H)|+1)
=\Bigl(\sum_{H\in\mathcal{P}} |E(H)|\Bigr)+|\mathcal{P}|
=m+r.
\]
Therefore,
\begin{equation}\label{eq:cost-vs-number-of-stars}
\cost(\mathcal{P})=(m+r)-n.
\end{equation}
In particular, $\cost(\mathcal{P})\le k+m-n$ holds if and only if $r\le k$.

\medskip
We now prove that $G$ has a star-partition into at most $k$ stars if and only if $G$ has a vertex cover of size at most $k$.

\smallskip
\noindent\textbf{(\(\Rightarrow\)).}
Assume that $E(G)$ admits a star-partition $\mathcal{P}$ with $|\mathcal{P}|\le k$.
For each star $H\in\mathcal{P}$, choose its center vertex and let $C$ be the set of all chosen centers (for a one-edge star, choose any of its endpoint as the center).
Since every edge of $G$ belongs to some star in $\mathcal{P}$ and every edge in a star is incident with the center of that star,
it follows that every edge of $G$ has at least one endpoint in $C$.
Thus, $C$ is a vertex cover of $G$.
Moreover, by construction, $|C|\le |\mathcal{P}|\le k$.

\smallskip
\noindent\textbf{(\(\Leftarrow\)).}
Conversely, assume that $G$ has a vertex cover $C$ with $|C|\le k$.
For each edge $e=uv\in E(G)$, at least one of $u$ and $v$ lies in $C$.
Fix an arbitrary assignment that maps each edge $e$ to one of its endpoints in $C$.
For every vertex $x\in C$, let $E_x$ be the set of edges assigned to $x$, and let $H_x$ be the star centered at $x$ with edge set $E_x$.
Then $\{H_x \colon  x\in C, E_x \neq \emptyset\}$ is a partition of $E(G)$ into stars, and its size is at most $|C|\le k$.

\smallskip
Combining the two directions, $G$ has a vertex cover of size at most $k$ if and only if $E(G)$ can be partitioned into at most $k$ stars.
Using~\eqref{eq:cost-vs-number-of-stars}, this holds if and only if $G$ admits a star-partition of total cost at most $k+m-n$.
Equivalently, $(G,k)$ is a \yes-instance of \VC\ if and only if $(G,k+m-n)$ is a \yes-instance of \SEPVC.
\end{proof}
}
%Since {\sc Vertex Cover} remains {\sf NP}-complete on graphs of girth $g$, for every fixed $g \geq 3$ \cite{yannakakis1981edge},  we have the following theorem.

Now, \Cref{thm:SEPVC-NPcomplete} follows from the hardness of \VC (\Cref{prop:vc-high-girth}), \Cref{lem:VC-to-SEPVC}, and the fact 
that $k+m-n\leq n+m-n = m$ (for the \VC instance, we can safely assume that $k\leq n$).

\begin{theorem}\label{thm:SEPVC-NPcomplete}
For every fixed integer $g\geq 3$, {\sc SEPVC} is
{\sf NP}-complete on subcubic graphs of girth at least $g$. Furthermore, assuming {\sf ETH},
it can be solved neither in time
$2^{o(n+m)}$ nor in time $2^{o(k)}\cdot n^{O(1)}$.
\end{theorem}

%\subsection{Putting all together}

% We have established polynomial-time equivalences ($\mathtt{poly}$ refers to polynomial-time many-one reducibility) between the three problems on connected graphs of girth at least five. Together with the {\sf NP}-hardness of {\sc SEPVC}, this immediately yields the complexity of Union of \textsc{Union of Stars-EVS} and \textsc{Cograph-EVS}.

Thus we obtain the following relations, under the restrictions mentioned in \Cref{lem:VC-to-SEPVC} and in \Cref{lem:SEPVC-equivalent-to-PFEVS}. The first follows from
\Cref{lem:VC-to-SEPVC}, while the second and the third follow from
\Cref{lem:SEPVC-equivalent-to-PFEVS} and
\Cref{lem:Pt-free-VS-equivalent-EVS-high-girth}, respectively, when restricted
to graphs of girth at least \(5\). Then \Cref{thm:SEPVC-NPcomplete} implies \Cref{thm:PFEVS-PFVS-NPcomplete}.

 \begin{center}
\tcbox[colback=gray!5,colframe=gray!50,boxrule=0.8pt,arc=2pt]
{\textsc{Vertex Cover} $\;\leq_{\mathtt{poly}}\;$ \textsc{SEPVC}$ \;\equiv\; $\PFEVS $\;\equiv\;$ \PFVS}
\end{center}

\section{\texorpdfstring{$\boldsymbol{P_5}$-free Vertex Splitting}{P5-free Vertex Splitting}}
\label{sec:p5}

In this section we prove \Cref{thm:P5VS-NPcomplete}. % A {\em double-star graph}  is a graph that is formed by two stars  via joining their centers by an edge. %We begin  with a simple observation.
First we formulate a graph partition problem that captures the effect of exclusive vertex splitting to $P_5$-free graphs.
Throughout this section, all input graphs are assumed to have no isolated vertices. 

\begin{definition}[Double-star edge partition]
  {\em   Let $G$ be a graph without any isolated vertices. A collection $\mathcal P=\{H_1,\ldots,H_r\}$ of subgraphs of $G$ is a
\emph{double-star edge partition} if
\begin{itemize}
\item each $H_i$ is a star or a double-star,
\item the $H_i$ are edge-disjoint,
\item $\bigcup_i E(H_i)=E(G)$.
\end{itemize}}
\end{definition}

The \emph{cost} of a double-star edge partition  $\mathcal{P}$ is denoted by $ \cost(\mathcal{P})$ where 
    \[
        \cost(\mathcal{P}) \coloneqq \sum_{v \in V(G)} \cost(v); ~~~~~\text{and}   ~~~~~\cost(v) \coloneqq \bigl|\{\, i \in [r] \mid v \in V(H_i) \,\}\bigr| - 1;
    \]

\begin{problem}
	\problemtitle{{\DSEPVC} ({\sc DSEPVC})}
	\probleminput{A graph $G$ without any isolated vertices, and a nonnegative integer $k$.}
	\problemquestion{Does $G$ admit a double-star edge partition of cost at most $k$?}
\end{problem}

We next show that  minimizing the number of exclusive vertex splits required to transform
a graph into $P_5$-free is equivalent to minimizing the vertex cost of
a double-star edge partition.

\begin{lemma}\label{lem:DSEPVC-equivalent-to-P5FEVS}
Let $G$ be a graph with girth at least $6$, %having at least $|V(G)|$ edges, 
and having no isolated vertices, and let $k$ be a
nonnegative integer. Then $(G,k)$ is a \yes-instance of \DSEPVC\ if and only if
$(G,k)$ is a \yes-instance of \PfiveFEVS.
\end{lemma}

\ifthenelse{\boolean{shortver}}{}{
\begin{proof}
\noindent\textbf{(\(\Rightarrow\)).}
Suppose $(G,k)$ is a \yes-instance of \DSEPVC. Let    $\mathcal{P}=\{H_1,\ldots,H_r\}$
be a partition of $E(G)$ into stars and double-stars with    $\cost(\mathcal{P})\le k$.
For each vertex $v\in V(G)$, let
\[
    p(v):=\bigl|\{\,i\in[r]\mid v\in V(H_i)\,\}\bigr|.
\]
We split $v$ into $p(v)$ copies, one copy for each part $H_i$ containing $v$.
This can be done using exactly $p(v)-1$ exclusive vertex splits at $v$.
After doing this for every vertex, each copy of $v$ is incident only with
edges belonging to one fixed part of $\mathcal{P}$.

Thus the resulting graph is the disjoint union of the graphs
$H_1,\ldots,H_r$. Each $H_i$ is a star or a double-star, and hence each
component has diameter at most $3$. Therefore the resulting graph is
$P_5$-free.

The total number of exclusive vertex splits used is
\[
    \sum_{v\in V(G)}(p(v)-1)
    =
    \cost(\mathcal{P})
    \le k.
\]
Hence $(G,k)$ is a \yes-instance of \PfiveFEVS.

\smallskip
\noindent\textbf{(\(\Leftarrow\)).}
Conversely, suppose $G$ can be transformed into a $P_5$-free graph $G'$ by at
most $k$ exclusive vertex splits.
Let $\phi:E(G)\to E(G')$ be the natural bijection between original edges and
their images after the exclusive splits. By
\Cref{obs:no-new-cycle}, the graph $G'$ has girth
at least $6$.

Every connected component of $G'$ is $P_5$-free and has girth at least $6$.
By \Cref{lem:forest-diameter-equivalent-to-Pdplus2-free-high-girth}, every nontrivial connected
component of $G'$ is a star or a double-star.

Let $\mathcal{C}$ be the set of nontrivial connected components of $G'$. For
each component $C\in\mathcal{C}$, define
\[
    F_C:=\{\,e\in E(G)\mid \phi(e)\in E(C)\,\}.
\]
Since $\phi$ is a bijection, the sets $F_C$, for $C\in\mathcal{C}$, form a
partition of $E(G)$.

We claim that each subgraph \(G[F_C]\) is a star or a double-star. Let
\(\pi\) denote the natural projection that maps every copy of a vertex in
\(G'\) to the corresponding original vertex of \(G\). Since \(C\) is a star
or a double-star, \(C\) is a tree of diameter at most \(3\).

We first show that \(\pi\) is injective on \(V(C)\). Suppose, for a
contradiction, that \(C\) contains two distinct vertices \(x\) and \(y\) such
that        $\pi(x)=\pi(y)=v$
for some \(v\in V(G)\). Since \(C\) has diameter at most \(3\), the unique
\(x\)-\(y\) path in \(C\) has length at most \(3\). Its length is not \(1\),
because distinct copies of the same original vertex are nonadjacent.

If this path has length \(2\), say \(xzy\), then the two distinct edges
\(xz\) and \(zy\) of \(C\) both project to the same edge
\(v\pi(z)\) of \(G\). This is impossible, since the edges of \(G'\) are in
natural bijection with the edges of \(G\).
Thus the \(x\)-\(y\) path has length \(3\), say        $xz_1z_2y$.
Then the projection of this path gives the cycle        $v,\ \pi(z_1),\ \pi(z_2),\ v$
in \(G\). This is a cycle of length \(3\), contradicting
\(\girth(G)\ge 6\). Therefore \(\pi\) is injective on \(V(C)\).

Since \(\pi\) is injective on \(V(C)\) and the edges of \(C\) are precisely
the images of the edges in \(F_C\), the graph \(G[F_C]\) is isomorphic to
\(C\). Hence \(G[F_C]\) is a star or a double-star.
Thus    $\mathcal{P}:=\{\,G[F_C]\mid C\in\mathcal{C}\,\}$
is a double-star edge partition of $G$.

For each vertex $v\in V(G)$, let $q(v)$ be the number of non-isolated copies
of $v$ in $G'$. Then $v$ belongs to exactly $q(v)$ members of $\mathcal{P}$.
If $v$ was split into $t(v)$ total copies in $G'$, then $q(v)\le t(v)$.
Therefore
\[
    \cost(\mathcal{P})
    =
    \sum_{v\in V(G)}(q(v)-1)
    \le
    \sum_{v\in V(G)}(t(v)-1).
\]
The right-hand side is precisely the number of exclusive vertex splits used to
obtain $G'$. Hence   $\cost(\mathcal{P})\le k$.
Thus $(G,k)$ is a \yes-instance of \DSEPVC.
\end{proof}
}

%\ifthenelse{\boolean{shortver}}{}{
\subsection{Hardness of \DSEPVC}
%}

%\paragraph{\EDS.}
Next we prove the hardness of {\sc DSEPVC} by giving a polynomial-time reduction from \EDS. 

An edge subset $F\subseteq E(G)$ is an \emph{edge dominating set} of $G$ if every
edge of $G$ either belongs to $F$ or shares an endpoint with an edge of $F$.
Given a graph $G$ and a nonnegative integer $k$, the problem \EDS\ asks whether  $G$ has an edge dominating set of
size at most $k$.
We prove the hardness of the required kind for \EDS (\Cref{lem:EDS-hard-high-girth}) from the known reductions (\Cref{prop:vc-to-eds} and \Cref{prop:eds-to-eds-plus-one}).

\begin{proposition}{\rm \cite[Lemma~2]{chlebik2006approximation}}
\label{prop:vc-to-eds}
Let $G$ be a graph without any isolated vertices, and $k$ be a nonnegative integer.
%Let $(G,k)$ be an instance of \VC. 
Let $G'$ be obtained from $G$ as follows.
Introduce a new vertex $z$ and add the edge $zu$ for every vertex
$u\in V(G)$. Moreover, for every edge $e=uv\in E(G)$, delete $e$ and add
new vertices $s_{uv}$, $s_e$, $s_{vu}$, and $s_{e'}$, together with the
edges
$
us_{uv}, s_{uv}s_e, s_es_{vu}, s_{vu}v, s_es_{e'}.
$
Then $(G,k)$ is a \yes-instance of \VC if and only if
$(G',k+|E(G)|)$ is a \yes-instance of \EDS.
\end{proposition}

\begin{proposition}{\rm \cite[Lemma~2]{chlebik2007complexity}}
\label{prop:eds-to-eds-plus-one}
%Let $(G,k)$ be an instance of \EDS. 
Let $G$ be a graph without any isolated vertices, and $k$ be a nonnegative integer.
Let $G'$ be obtained from $G$ by
replacing each edge by a path of length $4$, that is, by subdividing each
edge exactly three times. Then $(G,k)$ is a \yes-instance of \EDS if and
only if $(G',k+|E(G)|)$ is a \yes-instance of \EDS.
\end{proposition}

\begin{lemma}\label{lem:EDS-hard-high-girth}
For every fixed integer $g\geq 3$, \EDS
is \npc on graphs $G$ of girth at least $g$, satisfying $|E(G)|\geq |V(G)|$, and without any isolated vertices. 
Furthermore, assuming {\sf ETH}, 
it cannot be solved in time $2^{o(n+m)}$. 
\end{lemma}

\ifthenelse{\boolean{shortver}}{}{
\begin{proof}
We prove hardness by a reduction from
\VC on subcubic graphs, which is \textsf{NP}-complete and, assuming
\textsf{ETH}, cannot be solved in time $2^{o(n+m)}$ (\Cref{prop:vc-high-girth}). 
We assume that
the input graph has at least one edge and has no isolated vertices.

Let $(G,k)$ be an instance of \VC. Let $n=|V(G)|$ and $m=|E(G)|$. Construct
the graph $G'$ as in \Cref{prop:vc-to-eds}. Let $k' = k + |E(G)|$.
By \Cref{prop:vc-to-eds}, $(G,k)$ is a \yes-instance of \VC if and only if
$(G',k')$ is a \yes-instance of \EDS. Moreover, $|V(G')| = n+4m+1$ and $|E(G')| = n+5m$.
In particular, $|E(G')|-|V(G')| = m-1 \geq 0$, since we assumed that $m\geq 1$.

For a graph $H$, let $S(H)$ denote the graph obtained from $H$ by replacing
each edge by a path of length $4$. Define $S_0(G')=G'$ and, for $i\geq 1$,
define $S_i(G') = S(S_{i-1}(G'))$.
Thus every edge of $S_{i-1}(G')$ is replaced by a path of length $4$ in
$S_i(G')$. Hence $|E(S_i(G'))| = 4^i |E(G')|$ and
\[
|V(S_i(G'))|
=
|V(G')| + 3|E(G')|\sum_{j=0}^{i-1}4^j
=
|V(G')| + (4^i-1)|E(G')|.
\]
Consequently,
\[
|E(S_i(G'))|-|V(S_i(G'))|
=
4^i|E(G')|-\bigl(|V(G')|+(4^i-1)|E(G')|\bigr)
=
|E(G')|-|V(G')|
\geq 0.
\]

Let $i^\star$ be the smallest positive integer such that $g \leq 4^{i^\star}$. Let $G'' = S_{i^\star}(G')$.
Every cycle of $G''$ is obtained from a cycle of $G'$ by replacing each edge
of that cycle by a path of length $4^{i^\star}$. Since $G'$ is a simple graph,
its girth is at least $3$. Therefore $\girth(G'') = 4^{i^\star}\girth(G') \geq 4^{i^\star} \geq g$.
Also, by the calculation above, $|E(G'')|\geq |V(G'')|$.
Thus $G''$ satisfies the properties mentioned in the statement of the lemma.

By repeated applications of \Cref{prop:eds-to-eds-plus-one}, we have that
$(G',k')$ is a \yes-instance of \EDS if and only if $(G'',k'')$ is a
\yes-instance of \EDS, where
\[
k''
=
k' + |E(G')|\sum_{j=0}^{i^\star-1}4^j
=
k' + \frac{4^{i^\star}-1}{3}|E(G')|.
\]
Combining this with \Cref{prop:vc-to-eds}, we obtain that $(G,k)$ is a \yes-instance of \VC if and only if $(G'',k'')$
is a \yes-instance of \EDS.

It remains to verify that the reduction is linear for fixed $g$. Since
$i^\star$ is the smallest positive integer such that $g\leq 4^{i^\star}$,
we have $4^{i^\star}<4g$. Hence
\[
\begin{aligned}
|V(G'')|+|E(G'')|
=
|V(G')|+(4^{i^\star}-1)|E(G')|+4^{i^\star}|E(G')|
&\leq
|V(G')|+2\cdot 4^{i^\star}|E(G')|\\
&<
|V(G')|+8g|E(G')|.
\end{aligned}
\]

Thus the reduction is polynomial-time and linear-size for every fixed
integer $g\geq 3$.

Therefore \EDS is \npc on the restricted class. Moreover, if \EDS on the restricted class 
could be solved in time $2^{o(n+m)}$, then the above
linear-size reduction would yield a $2^{o(n+m)}$-time algorithm for \VC on
subcubic graphs, contradicting \textsf{ETH}. This proves the claimed lower
bound.
\end{proof}
}

\begin{lemma}\label{lem:EDS-to-DSEPVC}
Let $(G,k)$ be an instance of \EDS, where $G$ has girth at least
$4$, has at least $|V(G)|$ edges, and has no isolated vertices. %Let $n:=|V(G)|$ and $m:=|E(G)|$.
Then $(G,k)$ is a \yes-instance of \EDS\ if and only if $(G,k+m-n)$ is a
\yes-instance of \DSEPVC.
\end{lemma}

%\ifthenelse{\boolean{shortver}}{}{
\begin{proof}
Since $G$ has at least $|V(G)|$ edges, we have $m\ge n$, and hence $k+m-n$ is a nonnegative
integer.

We first relate the cost of a double-star edge partition to the number of its
parts. Let $\mathcal{P}$ be any partition of $E(G)$ into stars and
double-stars, and let $r:=|\mathcal{P}|$. For a vertex $v\in V(G)$, let $p(v):=\bigl|\{\,H\in\mathcal{P}\mid v\in V(H)\,\}\bigr|$.
Then
\[
    \cost(\mathcal{P})
    =
    \sum_{v\in V(G)}(p(v)-1)
    =
    \sum_{H\in\mathcal{P}} |V(H)| - n.
\]
Every part $H\in\mathcal{P}$ is a tree with at least one edge, and hence   $|V(H)|=|E(H)|+1$.
Therefore
\[
    \sum_{H\in\mathcal{P}} |V(H)|
    =
    \sum_{H\in\mathcal{P}} (|E(H)|+1)
    =
    m+r.
\]
Thus
\begin{equation}\label{eq:DSEPVC-cost}
    \cost(\mathcal{P}) = m+r-n.
\end{equation}
Consequently,   $\cost(\mathcal{P})\le k+m-n    \quad\Longleftrightarrow\quad    r\le k.$

It remains to prove that $G$ has an edge dominating set of size at most $k$ if
and only if $E(G)$ can be partitioned into at most $k$ stars and double-stars.

\smallskip
\noindent\textbf{(\(\Rightarrow\)).}
Suppose $F\subseteq E(G)$ is an edge dominating set with $|F|\le k$.

For every edge $e\in F$, we create one part $H_e$. Since $F$ is edge
dominating, every edge $f\in E(G)$ is dominated by at least one edge of $F$.
That is, either $f\in F$, or $f$ shares an endpoint with some edge of $F$.

Assign each edge $f\in E(G)$ to one edge $e\in F$ that dominates it. If
$f\in F$, assign $f$ to itself. For every $e=xy\in F$, let $H_e$ be the
subgraph of $G$ whose edge set consists exactly of the edges assigned to $e$.

We claim that each nonempty $H_e$ is a star or a double-star. Indeed, every
edge assigned to $e=xy$ shares an endpoint with $e$, and $e$ itself is assigned
to $H_e$. Hence every edge of $H_e$ is incident with $x$ or with $y$, and the
edge $xy$ belongs to $H_e$. Since $\girth(G)\ge 4$, no vertex can be adjacent
to both $x$ and $y$; otherwise there would be a triangle. Therefore the edges
of $H_e$ form a tree of diameter at most $3$, with possible centers $x$ and
$y$. Hence $H_e$ is a star or a double-star.

The graphs $H_e$, for $e\in F$, are pairwise edge-disjoint and together cover
$E(G)$, because every edge of $G$ was assigned to exactly one edge of $F$.
Thus they form a partition of $E(G)$ into at most $|F|\le k$ stars and
double-stars.

By \eqref{eq:DSEPVC-cost}, this partition has cost at most $k+m-n$.
Therefore $(G,k+m-n)$ is a \yes-instance of \DSEPVC.

\smallskip
\noindent\textbf{(\(\Leftarrow\)).}
Conversely, suppose $(G,k+m-n)$ is a \yes-instance of \DSEPVC. Let    $\mathcal{P}=\{H_1,\ldots,H_r\}$
be a partition of $E(G)$ into stars and double-stars with    $\cost(\mathcal{P})\le k+m-n$.
By \eqref{eq:DSEPVC-cost}, we have $r\le k$.

For each part $H_i$, choose an edge $e_i\in E(H_i)$ as follows. If $H_i$ is a
double-star, choose its central edge. If $H_i$ is a star, choose any edge
incident with its center. Define    $F:=\{e_i\mid i\in[r]\}$.
We claim that $F$ is an edge dominating set of $G$.

Let $f\in E(G)$ be arbitrary. Since $\mathcal{P}$ partitions $E(G)$, the edge
$f$ belongs to some part $H_i$. If $H_i$ is a star, then every edge of $H_i$
shares the center of the star with the chosen edge $e_i$. Hence $f$ is
dominated by $e_i$. If $H_i$ is a double-star, then every edge of $H_i$ shares
an endpoint with the central edge, and $e_i$ was chosen to be this central
edge. Hence again $f$ is dominated by $e_i$.

Therefore every edge of $G$ is dominated by some edge of $F$. Thus $F$ is an
edge dominating set. Moreover,    $|F| = r\le k$.
Hence $(G,k)$ is a \yes-instance of \EDS.
This completes the proof.
\end{proof}
%}

Now \Cref{thm:DSEPVC-NPcomplete} follows from \Cref{lem:EDS-hard-high-girth} and \Cref{lem:EDS-to-DSEPVC}.
The parameterized lower bound follows from the fact that $k$ can be assumed to be at most $m$ (as $(G,k)$ is an instance of \EDS)
and hence $k+m-n\leq 2m$.

\begin{theorem}\label{thm:DSEPVC-NPcomplete}
For every fixed integer $g\geq 4$, {\sc DSEPVC}
is \npc on the class of
graphs $G$ of girth at least $g$, satisfying $|E(G)|\geq |V(G)|$, and having no isolated vertices. 
Furthermore, assuming {\sf ETH}, 
it can be solved neither in time $2^{o(n+m)}$ nor in time $2^{o(k)}\cdot n^{O(1)}$.
\end{theorem}

So  we have established the following relations under the restrictions mentioned in \Cref{lem:EDS-to-DSEPVC}, \Cref{lem:DSEPVC-equivalent-to-P5FEVS}, and \Cref{lem:Pt-free-VS-equivalent-EVS-high-girth}. The first follows from \Cref{lem:EDS-to-DSEPVC}, the second from \Cref{lem:DSEPVC-equivalent-to-P5FEVS},
and the third from \Cref{lem:Pt-free-VS-equivalent-EVS-high-girth}. Then \Cref{thm:DSEPVC-NPcomplete} proves \Cref{thm:P5VS-NPcomplete}.

 \begin{center}
\tcbox[colback=gray!5,colframe=gray!50,boxrule=0.8pt,arc=2pt]
{\EDS$ \;\leq_{\mathtt{poly}}\; \textsc{DSEPVC}  \;\equiv\;$  \PfiveFEVS$\;\equiv\; $ \PfiveFVS}
\end{center}

\section{\texorpdfstring{$\boldsymbol{P_t}$-free Vertex Splitting for $\boldsymbol{t\geq 6}$}{Pt-free Vertex Splitting for t >= 6}}
\label{sec:pt}

In this section, we prove the hardness for {\sc $P_t$-free-VS} and variants, for every fixed integer $t \geq 6$.
We reduce from \textsc{Vertex Cover} on triangle-free graphs. For every
edge $uv$ of the input graph, the constructed graph contains an induced
$P_t$ passing through the vertices corresponding to $u$ and $v$.
This is obtained by attaching pendant paths of adequate lengths to each vertex.
Therefore, any solution using at most $k$ vertex splits yields a vertex
cover of size at most $k$. Conversely, given a vertex cover $C$, we split
each vertex of $C$ so that its two pendant paths form a separate component.
The remaining core is $P_5$-free, and the lengths of the attached paths
ensure that no induced $P_t$ is created. Thus the input graph has a vertex
cover of size at most $k$ if and only if the constructed graph can be made
$P_t$-free using at most $k$ splits. The same construction uses exclusive
and shallow splits, so it also applies to the corresponding
\textsc{EVS}, \textsc{SVS}, and \textsc{SEVS} variants.

\ifthenelse{\boolean{shortver}}{}{
Now we give the formal proof.
}
\theorempt*

\ifthenelse{\boolean{shortver}}{}{
\begin{proof}
We reduce from \textsc{Vertex Cover} on triangle-free graphs, which is
\textsf{NP}-hard by \Cref{prop:vc-high-girth}. Let
\[
a=\Big\lfloor \frac{t-4}{2}\Big\rfloor,
\qquad
b=\Big\lceil \frac{t-4}{2}\Big\rceil.
\]
Then $a+b=t-4$ and $1\le a\le b$.
Since $t\ge 6$,
$
b+4<t
~\text{and}~
2b+2<t.
$

Before constructing the gadget, we prove a claim.

\begin{claim}
\label{claim:attachment}
Let $K$ be a $P_5$-free graph and let $X$ be a clique in $K$.
Attach to every vertex of $X$ two pendant paths of lengths
$a$ and $b$, where $a+b=t-4$ and $b=\max\{a,b\}$.
Then the resulting graph is $P_t$-free.
\end{claim}

\begin{claimproof}
Let $P$ be an induced path in the resulting graph. Since $K$ is $P_5$-free,
the intersection $P\cap K$ contains at most four vertices. The path $P$ can use vertices from at most two pendant paths. Now

\begin{itemize}
    \item If $P$ uses vertices from only one pendant path, then
$
|V(P)|\le b+4<t.
$

\item Suppose that $P$ uses vertices from two pendant paths with
attachment vertices $x,x'\in X$.
Since $X$ is a clique, $xx'\in E(K)$. If the section of $P$ contained in $K$ between $x$ and $x'$
contains another vertex, then $x$ and $x'$ occur
non-consecutively on $P$ while remaining adjacent,
creating a chord.
Hence $x$ and $x'$ must be consecutive on $P$. Therefore the part of $P$ outside the two pendant paths
consists of exactly the vertices $x$ and $x'$.
Consequently,
$
|V(P)|\le b+b+2=2b+2<t.
$
\end{itemize}

 Thus no induced $P_t$ exists.
\end{claimproof}

Let $(G,k)$ be an instance of {\sc Vertex Cover},
where $G$ is triangle-free. Since a house (complement of $P_5$) contains a triangle, $G$ is house-free. We construct an instance $(G_t,k)$ of {\sc $P_t$-free-VS} as follows.

\begin{itemize}
    \item Let $O:=V(G)$. Add $O$ to $G_t$. Then 
for every pair of distinct vertices $x,y\in O$,
add the edge $xy$ to $G_t$ if and only if
$xy\notin E(G)$.
Thus
$
G_t[O]\cong \overline{G}.
$

 \item For every vertex $v\in O$, attach two pendant paths,
$
X_v=x^v_1x^v_2\cdots x^v_a$ and 
$
Y_v=y^v_1y^v_2\cdots y^v_b,
$
by adding the edges
$
vx^v_1,\;
x^v_1x^v_2,\ldots,
x^v_{a-1}x^v_a
$
and
$
vy^v_1,\;
y^v_1y^v_2,\ldots,
y^v_{b-1}y^v_b.
$

 \item Next for every edge $uv\in E(G)$,
introduce two vertices
$s_{uv}$ and $s_{vu}$.
Let $S$ denote the set of all such vertices. Make $S$ a clique.

\item Furthermore, for every  edge $uv\in E(G)$,
make $s_{uv}$ adjacent to every vertex of $O$
except $v$, and make $s_{vu}$ adjacent to every
vertex of $O$ except $u$. 
\end{itemize}

This completes the construction. See \Cref{fig:gadget} for an illustration. Clearly this construction takes polynomial time.
\begin{figure}[ht!]
\centering
\begin{tikzpicture}[
    scale=.7,
    every node/.style={font=\small},
    O/.style={circle,draw,thick,fill=blue!20,minimum size=5mm},
    S/.style={circle,draw,thick,fill=red!25,minimum size=5mm},
    P/.style={circle,draw,thick,fill=green!20,minimum size=4mm}
]

% Original vertices
\node[O] (u) at (0,0) {$u$};
\node[O] (v) at (7,0) {$v$};

% S-vertices
\node[S] (suv) at (2,2) {$s_{uv}$};
\node[S] (svu) at (5,2) {$s_{vu}$};

% Edges involving S-vertices
\draw[thick] (suv)--(svu);
\draw[thick] (u)--(suv);
\draw[thick] (v)--(svu);

% Non-edges
\draw[dashed] (u)--(svu);
\draw[dashed] (v)--(suv);

% Path X_u attached to u
\node[P] (xu1) at (-1,-1) {};
\node[P] (xu2) at (-2,-2) {};
\draw[thick] (u)--(xu1);
\draw[thick] (xu1)--(xu2);
\node at (-2.4,-2.6) {$X_u$};

% Path Y_u attached to u
\node[P] (yu1) at (1,-1) {};
\node[P] (yu2) at (2,-2) {};
\draw[thick] (u)--(yu1);
\draw[thick] (yu1)--(yu2);
\node at (2.4,-2.6) {$Y_u$};

% Path X_v attached to v
\node[P] (xv1) at (6,-1) {};
\node[P] (xv2) at (5,-2) {};
\draw[thick] (v)--(xv1);
\draw[thick] (xv1)--(xv2);
\node at (4.6,-2.6) {$X_v$};

% Path Y_v attached to v
\node[P] (yv1) at (8,-1) {};
\node[P] (yv2) at (9,-2) {};
\draw[thick] (v)--(yv1);
\draw[thick] (yv1)--(yv2);
\node at (9.4,-2.6) {$Y_v$};

\end{tikzpicture}
\caption{The vertices and edges in $G_t$ corresponding to an edge
$uv\in E(G)$. Here $t=8$. Dashed lines denote non-edges. The set $S$ forms a clique,
and each $s_{xy}$ is adjacent to every original vertex except $y$.}
\label{fig:gadget}
\end{figure}
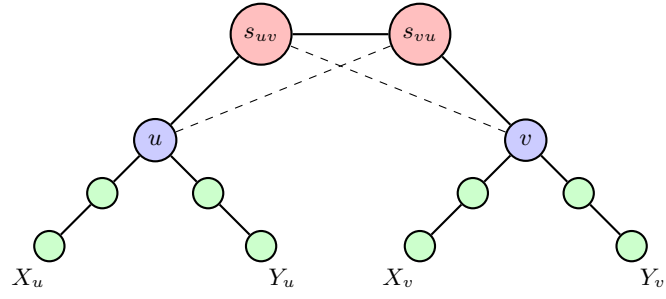

\begin{claim}
 $(G,k)$ is a \yes-instance of
{\sc Vertex Cover}
if and only if $(G_t,k)$ is a \yes-instance of {\sc $P_t$$\mbox{-}$free-VS}  (and of {\sc EVS}/{\sc SVS}/{\sc SEVS}).

\end{claim}

\begin{claimproof}
    $(\Leftarrow)$ Assume that $(G_t,k)$ is a \yes-instance of {\sc $P_t$$\mbox{-}$free-VS}, i.e., $G_t$ can be transformed into a
$P_t$-free graph using at most $k$ vertex splits. Consider an arbitrary edge $uv\in E(G)$.
By construction,
$
P_{uv}:=
x^u_a,x^u_{a-1},\ldots,x^u_1,
u,s_{uv},s_{vu},v,
y^v_1,\ldots,y^v_b
$
is an induced path on exactly
$
a+b+4=t
$
vertices. Hence every induced copy of $P_t$ of this form must be
destroyed by the splitting process.
Consequently, for every edge $uv\in E(G)$,
at least one vertex of $P_{uv}$ is split. We construct a set $C\subseteq V(G)$ as follows.

\begin{itemize}
\item If an original vertex $v\in O$ is split,
      add $v$ to $C$.
\item If a split vertex lies on one of the two pendant
      paths attached to $v$, add $v$ to $C$.
\item If the split vertex is $s_{uv}$ or $s_{vu}$,
      add one of $u$ and $v$ (chosen arbitrarily) to $C$.
\end{itemize}

Each split contributes at most one vertex to $C$,
and therefore $|C|\le k$. Now let $uv\in E(G)$.
Since the induced path $P_{uv}$ must be destroyed,
some vertex of $P_{uv}$ is split.
By the construction of $C$,
either $u$ or $v$ is added to $C$.
Thus every edge of $G$ is incident with a vertex of $C$,
so $C$ is a vertex cover of $G$.
Hence $(G,k)$ is a \yes-instance of
{\sc Vertex Cover}.
\smallskip

  \noindent $(\Rightarrow)$ Assume that $G$ admits a vertex cover $C$ with $|C|\le k$. For every $v\in C$, split $v$ into two vertices
$v_T$ and $v_R$.
The vertex $v_T$ is made adjacent only to
$x^v_1$ and $y^v_1$,
while $v_R$ inherits all remaining neighbours of $v$. Let $H$ denote the resulting graph and let
$
I:=V(G)\setminus C.
$

\smallskip 
Since $C$ is a vertex cover,
$I$ is an independent set in $G$.
Therefore $I$ is a clique in $\overline G$. Let 
$
R:=I\cup \{\,v_R : v\in C\,\}.
$
Then
$
H[R]\cong \overline G.
$ Since $G$ is house-free, the graph $\overline G$
contains no induced $P_5$.
Hence $H[R]$ is $P_5$-free. 

\smallskip
We next show that $H[R\cup S]$ is $P_5$-free. Suppose, for contradiction, that $Q$ is an induced $P_5$ in
$H[R\cup S]$. Since $H[R]$ is $P_5$-free, the path $Q$ contains at least one
vertex of $S$.
Since $S$ is a clique, $Q$ contains at most two vertices of $S$,
and if it contains two, then they are consecutive on $Q$. Observe that every vertex of $S$ has exactly one non-neighbour in
$R$. Indeed, for every edge $uv\in E(G)$, the vertex $s_{uv}$ is
non-adjacent only to the copy of $v$, namely $v$ if $v\in I$ and
$v_R$ if $v\in C$. First assume that $Q$ contains exactly one vertex $s\in S$.
Then the other four vertices of $Q$ belong to $R$.
Since $s$ has exactly one non-neighbour in $R$, it is adjacent to
at least three of these four vertices.
Hence $s$ has at least three neighbours in $Q$.
However, every vertex of an induced $P_5$ has degree at most two
within the path, a contradiction. Now assume that $Q$ contains exactly two vertices of $S$.
Let these vertices be $s$ and $s'$.
As noted above, they are consecutive on $Q$.
The remaining three vertices of $Q$ belong to $R$.
Since each of $s$ and $s'$ has exactly one non-neighbour in $R$,
each is adjacent to at least two of these three vertices.
Together with the edge $ss'$, it follows that each of $s$ and $s'$
has at least three neighbours in $Q$.
Again, this contradicts that $Q$ is an induced $P_5$. Therefore $Q$ cannot exist, and hence $H[R\cup S]$ is $P_5$-free.

\smallskip 

For every $v\in C$, the vertex $v_T$ together with its
two pendant paths induces a connected graph on
$
a+b+1=t-3
$
vertices.
Consequently, no induced $P_t$ can be contained in such
a component.
It remains to deal with the pendant paths attached to vertices of $I$. We have shown that $H[R\cup S]$ is $P_5$-free. Now applying Claim~\ref{claim:attachment}
with
$
K=H[R\cup S]
$ and $
X=I
$,
we conclude that $H$ is $P_t$-free. Hence $(G_t,k)$ is a \yes-instance of
{\sc $P_t$-free-VS}. The statement for \textsc{$P_t$-free-EVS/SVS/SEVS} follows because all splits used in the construction
are exclusive and shallow.\end{claimproof}

Now, \Cref{thm:pt} follows from \Cref{prop:vc-high-girth} and the fact that $|V(G_t)| = O(|V(G)|+|E(G)|)$. 
\end{proof}
}

\section{Chordal Vertex Splitting}
\label{sec:chordal}
In this section, we prove that \chovs and its variants are \npc, and assuming ETH, there exists no subexponential-time algorithm for the problem. 
We use the \textsc{Chain Vertex Deletion} problem as our starting problem for the reduction. Recall that a bipartite graph is a \emph{chain graph} if and only if it has
no induced $2K_2$~\cite[Lemma 1]{yannakakis1981computing}.  We define the problem as follows.
\begin{problem}
	\problemtitle{{\textsc{Chain Vertex Deletion}} ({\sc CVD})}
	\probleminput{A graph $G$ and a nonnegative integer $k$.}
	\problemquestion{Does there exist a set $S\subseteq V(G)$ with $|S|\leq k$ such that $G-S$ is a
chain graph?}
\end{problem}

Yannakakis~\cite{DBLP:journals/siamcomp/Yannakakis81a} proved the NP-completeness of {\sc CVD} by a reduction from \textsc{3-SAT}.
The reduction generates a graph with number of vertices linear in the number of variables and the number of clauses of the \textsc{3-SAT}
instance. Thus we get \Cref{prop:cvd}.
%We use the following result about \textsc{Chain Vertex Deletion}.

\begin{proposition}{\rm \cite[Lemma~4]{DBLP:journals/siamcomp/Yannakakis81a}}
    \label{prop:cvd}
    {\sc Chain Vertex Deletion} is \npc on bipartite graphs. Moreover, assuming {\sf ETH}, the problem cannot be solved in time $2^{o(n)}$.
\end{proposition}

Since we can assume that $k\leq n$ for {\sc Chain Vertex Deletion}, a lower bound of $2^{o(n)}$ implies a lower bound of $2^{o(k)}\cdot n^{O(1)}$.
Now we are ready to prove the hardness of \chovs.

To prove the hardness for \chovs, we reduce from \textsc{Chain Vertex Deletion}. Given a bipartite graph
\(B=(X,Y;E_B)\), construct \(H\) by making both \(X\) and \(Y\) cliques and
keeping the same edges between \(X\) and \(Y\). Then an induced \(2K_2\) in
\(B\) corresponds exactly to an induced \(C_4\) in \(H\); hence \(H\) is
chordal if and only if \(B\) is a chain graph.
If \(S\) is a chain vertex deletion set of \(B\), then \(H-S\) is chordal.
Instead of deleting vertices of \(S\), we split each of them once, separating
its \(X\)-side and \(Y\)-side adjacencies. The new copies are simplicial over
cliques, and after eliminating them the remaining graph is \(H-S\). Thus
\(H\) can be made chordal using at most \(|S|\) splits.
Conversely, if at most \(k\) splits make \(H\) chordal, let \(S\) be the set
of split vertices. If \(B-S\) contained an induced \(2K_2\), the corresponding
four unsplit vertices would still induce a \(C_4\) in the final graph, a
contradiction. Hence \(B-S\) is a chain graph.

\ifthenelse{\boolean{shortver}}{}{Now we give the formal proof.}

\theoremchordal*

%\ifthenelse{\boolean{shortver}}{}{
\begin{proof}
We prove the statement for \chovs. The statements for the variants follow from the fact that
all our splits are shallow and exclusive.
We give a polynomial-time reduction from
\textsc{Chain Vertex Deletion}. Let $(B,k)$ be an instance of
\textsc{Chain Vertex Deletion}, where $B=(X,Y;E_B)$
is bipartite with bipartition $(X, Y)$ and 
$E_B$ is the set of edges.
We construct a graph $H$ from $B$ as follows. The graph $H$ contains all the vertices and edges in $B$. 
Additionally, the set $X$ forms a clique and the set $Y$ forms a clique in $H$.
%The vertex set of $H$ is $V(H)=X\cup Y$.
%We make $X$ a clique and $Y$ a clique. For every $x\in X$ and $y\in Y$, we
%add the edge $xy$ in $H$ if and only if $xy\in E_B$. 
%We first observe the following about the constructed graph~$H$. 
%This is the same construction used by Yannakakis~\cite{yannakakis1981computing} for a reduction from \textsc{Chain Completion}
%to \textsc{Chordal Completion (Minimum Fill-In)}.
The following observation about the constructed graph~$H$ follows from the work of Yannakakis~\cite{yannakakis1981computing}, who used
it to show a reduction from \textsc{Chain Completion}
to \textsc{Chordal Completion (Minimum Fill-In)}.

\begin{claim}{\rm \cite[Lemma~2]{yannakakis1981computing}}
\label{clm:chain-chordal}
$H$ is chordal if and only if $B$ is a chain graph.
\end{claim}%\skr{I think we can remove this proof and cite~\cite[Lemma 2]{yannakakis1981computing}. If we agree with this then we just have to uncomment the lines that I've written above.}

% \begin{proof}[Proof of the claim]
% Since both $X$ and $Y$ are cliques in $H$, every induced cycle of length at
% least four in $H$ contains at most two vertices from $X$ and at most two
% vertices from $Y$. Hence every hole of $H$, if one exists, is an induced
% $C_4$ with two vertices in $X$ and two vertices in $Y$.

% Let $x_1,x_2\in X$ and $y_1,y_2\in Y$. The vertices
% $\{x_1,x_2,y_1,y_2\}$ induce a $C_4$ in $H$ precisely when, up to relabelling, $x_1y_1,\ x_2y_2\in E(H)$ and $x_1y_2,\ x_2y_1\notin E(H)$.
% By the definition of $H$, this is equivalent to $x_1y_1,\ x_2y_2\in E_B$ and $x_1y_2,\ x_2y_1\notin E_B$.
% Thus $\{x_1,x_2,y_1,y_2\}$ induces a $2K_2$ in $B$.

% Therefore $H$ has a hole if and only if $B$ has an induced $2K_2$. Since a
% bipartite graph is a chain graph if and only if it is $2K_2$-free, the claim
% follows.
% \end{proof}

% \begin{proof}[Proof of the claim]
% Since $X$ and $Y$ are cliques in $H$, every induced cycle of length at least
% four uses at most two vertices from each side and therefore is an induced
% $C_4$. Consequently, $H$ is chordal if and only if it contains no induced
% $C_4$. For any $x_1,x_2\in X$ and $y_1,y_2\in Y$, the graph $H[\{x_1,x_2,y_1,y_2\}]$ is an induced $C_4$ exactly when $B[\{x_1,x_2,y_1,y_2\}]$ is an induced $2K_2$. Thus $H$ contains an induced
% $C_4$ if and only if $B$ contains an induced $2K_2$.
% Since $B$ is bipartite, it is a chain graph if and only if it is
% $2K_2$-free. Therefore $H$ is chordal if and only if $B$ is a chain graph.
% \end{proof}

We now prove the correctness of the reduction, i.e., we show that~$(B, k)$ is a \yes-instance of \textsc{Chain Vertex Deletion} if and only if~$(H, k)$ is a \yes-instance of \chovs.

\medskip

 $(\Rightarrow).$ Suppose $(B,k)$ is a \yes-instance of \textsc{Chain Vertex Deletion}. Let
$S\subseteq V(B)$ be such that $|S|\leq k$ and $B-S$ is a chain graph. Let~$S_X:=S\cap X\ \text{and}\ S_Y:=S\cap Y$.
By the above claim, $H-S$ is chordal.
We show that $H$ can be made chordal using at most $|S|$ vertex splits. In fact, we
split every vertex of $S$ exactly once to make the resulting graph chordal.

First, split each $x\in S_X$ into two non-adjacent vertices
$x^X$ and $x^Y$. We then define the adjacency of $x^X$ and $x^Y$ as follows:
\[
N(x^X)=(X \setminus S_X) \cup \{v^X \colon v \in S_X \setminus \{x\}\},
\qquad
N(x^Y)=N_H(x)\cap Y.
\] 

Observe that this is a valid vertex split since every original edge incident with $x$ is represented by
exactly one edge incident with $x^X$ or $x^Y$ in the final graph.
%$N_H(x)$ is the (disjoint) union of the two sets defined above.

After all vertices of $S_X$ have been split, we split each vertex $y\in S_Y$ into two non-adjacent vertices $y^Y$ and $y^X$ as follows.
The vertex $y^Y$ receives all neighbors of $y$ on the $Y$-side, and also
receives every already-created vertex $x^Y$ such that $xy\in E(H)$, where
$x\in S_X$. The vertex $y^X$ receives the remaining neighbors of $y$, namely
the neighbors of $y$ in $X\setminus S_X$. %Again this is a valid split.
%We similarly split each vertex~$y \in S_Y$ into two non-adjacent vertices~$y^Y$ and~$y^X$ and define the adjacencies as follows:
Formally, we define:
\[
\begin{aligned}
N(y^Y) &= (Y \setminus S_Y)
\cup \{v^Y \colon v \in S_Y \setminus \{y\}\} \cup \{x^Y \colon x \in N_H(y) \cap S_X\}, \text{and}\\
N(y^X) &= N_H(y) \cap (X \setminus S_X)
%\cup \{x^X \colon \text{for each } x \in N_H(y) \cap S_X\}
.
\end{aligned}
\]

Again, it is clear that the split defined above is a valid split. Let $\widehat{H}$ be the graph obtained after all these splits. We prove that
$\widehat{H}$ is chordal by giving a perfect elimination ordering.
Recall that a perfect elimination ordering of a chordal graph $J$ is an
ordering $v_1, v_2,\ldots, v_{|V(J)|}$ of its vertices such that for every
vertex $v_i$, the neighbors of $v_i$ with higher indices form a clique, i.e., 
$v_i$ is a simplicial vertex (a vertex whose neighborhood forms a clique) 
in the graph induced by $\{v_i, v_{i+1}, \ldots, v_{|V(J)|}\}$. It is well-known 
that a graph is chordal if and only if it admits a perfect elimination ordering~\cite{Rose70, RoseTL76}.

First consider a vertex $x^Y$ with $x\in S_X$. Its neighborhood in
$\widehat{H}$ is contained in $(Y\setminus S_Y)\cup \{y^Y : y\in S_Y\}$.
This set is a clique. Indeed, by construction of~$\widehat{H}$, the set~$Y\setminus S_Y$ is a clique, the set~$\{y^Y \colon y \in S_Y\}$ forms a clique, and each $y^Y$ is adjacent to every vertex of~$Y\setminus S_Y$. Hence every vertex $x^Y$, $x\in S_X$, is simplicial.
Similarly, for each $y\in S_Y$, the vertex $y^X$ has neighborhood contained
in $X\setminus S_X$, which is a clique. Hence every such $y^X$ is simplicial.
We eliminate first all vertices of the form $x^Y$, $x\in S_X$, and all
vertices of the form $y^X$, $y\in S_Y$.

After deleting these simplicial vertices, the remaining graph consists of
$H-S$, together with vertices $x^X$ for $x\in S_X$ whose whole neighborhood is the clique $(X\setminus S_X)\cup \{v^X : v\in S_X\setminus \{x\}\}$,
and vertices $y^Y$ for $y\in S_Y$ whose whole neighborhood 
%in~$H-S$ 
is the clique $(Y\setminus S_Y)\cup \{v^Y : v\in S_Y\setminus\{y\}\}$.
Since $H-S$ is chordal, and the remaining vertices $x^X$ and $y^Y$ are
simplicial over cliques of $H-S$ together with previously introduced
simplicial vertices, the whole graph $\widehat{H}$ is chordal.

Thus $H$ can be made chordal using exactly $|S|\leq k$ vertex splits.
Therefore $(H,k)$ is a \yes-instance of \chovs.

 $(\Leftarrow).$
Suppose $(H,k)$ is a \yes-instance of \chovs. Then
there is a sequence of at most $k$ vertex splits that transforms $H$ into a
chordal graph. Let $S=\{v\in V(H): v \text{ is split at least once}\}$. Clearly, $|S|\leq k$.
We claim that $B-S$ is a chain graph. Suppose not. Then $B-S$ contains an
induced $2K_2$. Hence there exist vertices $x_1,x_2\in X\setminus S$ and $y_1,y_2\in Y\setminus S$
such that %, up to relabelling, 
$x_1y_1,\ x_2y_2\in E_B$ and $x_1y_2,\ x_2y_1\notin E_B$.
By the construction of $H$, this means $x_1y_1,\ x_2y_2\in E(H)$
and $x_1y_2,\ x_2y_1\notin E(H)$.
Moreover, since $X$ and $Y$ are cliques in $H$, we have $x_1x_2\in E(H)$ and $y_1y_2\in E(H)$.
Therefore the four vertices $\{x_1, y_1, y_2, x_2\}$ induce a $C_4$ in %$H$, which is a contradiction.
$H-S$. Since we cannot change the adjacency (or non-adjacency) between vertices of~$H-S$ by splitting vertices of~$S$, the vertices~$\{x_1, y_1, y_2, x_2\}$ induce a~$C_4$ even after splitting the vertices of~$S$, %contradicting the fact that~$S$ is a solution to~\chovs.
contradicting that the final graph obtained by the split sequence is chordal.

%None of these four vertices belongs to $S$, so none of them is split. Vertex
%splitting vertices outside this set cannot delete edges among these four
%vertices and cannot add chords between them. Therefore the same four vertices
%still induce a $C_4$ in the final graph, contradicting the assumption that
%the final graph is chordal.

Hence $B-S$ is a chain graph. Since $|S|\leq k$, the instance $(B,k)$ is a
\yes-instance of \textsc{Chain Vertex Deletion}. Now the statements follow from \Cref{prop:cvd}. 
In particular, the {\sf ETH}-based lower bound uses the fact that the number of vertices in $H$
is the same as that of $B$, and the fact that the parameter is unchanged. This completes the proof.
\end{proof}
%}
\section{Unit Interval Vertex Splitting}
\label{sec:unit-interval}
In this section we prove the hardness of \uivs and its shallow variant.
We first recall some basic definitions from \cite{ardevol2025unitmultiple}. Given a family $\mathcal{F}$ of
unit-length closed intervals on the real line, the intersection graph
$G=\Omega(\mathcal{F})$ of $\mathcal{F}$ is defined as follows: for each
interval in $\mathcal{F}$, there is a corresponding vertex in $G$, and
two vertices of $G$ are adjacent if and only if the corresponding
intervals intersect. A graph $G$ is a \emph{unit interval graph} if there
is a family $\mathcal{F}$ of unit-length intervals such that
$\Omega(\mathcal{F})=G$.

A \emph{unit multiple-interval} is a union of pairwise disjoint unit
intervals. The intersection graph of a family of unit multiple-intervals
is defined analogously. A graph $G$ is a \emph{unit multiple-interval
graph} if there is a family $\mathcal{F}$ of unit multiple-intervals such
that $\Omega(\mathcal{F})=G$. Clearly, unit interval graphs form a
subclass of unit multiple-interval graphs.

A colored graph is a pair $(G,\gamma)$, where
$\gamma:V(G)\rightarrow \{\mathsf{black},\mathsf{white}\}$. A colored
graph $(G,\gamma)$ is a \emph{colored unit $2$-interval graph} if it has
a unit multiple-interval representation satisfying the following
conditions.
\begin{itemize}
    \item Each $\black$ vertex is represented by one unit interval.
    \item Each $\white$ vertex is represented by a union of two disjoint unit
    intervals.
\end{itemize}

Such a representation is called a \emph{colored unit $2$-interval
representation}.

\begin{problem}
    \problemtitle{{\cutir}}
    \probleminput{A colored graph $(G,\gamma)$.}
    \problemquestion{Decide whether $(G,\gamma)$ is a colored unit
    $2$-interval graph.}
\end{problem}

Given a colored graph $(G,\gamma)$, we say that a pair $(S,f)$, where
$S$ is a graph and $f:V(S)\rightarrow V(G)$ is a function, is a
\emph{split} of $(G,\gamma)$ if the following conditions hold.
\begin{enumerate}
    \item $|f^{-1}(v)|=1$ for every $v\in V(G)$ with
    $\gamma(v)=\mathsf{black}$.
    \item $|f^{-1}(v)|=2$ for every $v\in V(G)$ with
    $\gamma(v)=\mathsf{white}$.
    \item For every vertex $v$ of $G$, the set $f^{-1}(v)$ is an
    independent set in $S$.
    \item For every edge $st\in E(S)$, we have $f(s)f(t)\in E(G)$.
    \item For every edge $uv\in E(G)$, there exist
    $s\in f^{-1}(u)$ and $t\in f^{-1}(v)$ such that $st\in E(S)$.
\end{enumerate}

\begin{proposition}{\rm \cite[Lemma~6]{ardevol2025unitmultiple}}
\label{prop:ui}
A colored graph $(G,\gamma)$ is a colored unit $2$-interval graph if and
only if there exists a unit interval graph $S$ and a function
$f:V(S)\rightarrow V(G)$ such that $(S,f)$ is a split of $(G,\gamma)$.
\end{proposition}

Mart{\'i}nez, Rizzi, Sikora, and Vialette~\cite{ardevol2025unitmultiple}
prove the \textsf{NP}-completeness of \cutir\ by a reduction from the
restricted variant of \textsc{3SAT} stated below. The hardness of the latter
follows from a standard reduction from \textsc{3SAT}.

%Mart{\'i}nez, Rizzi, Sikora, and Vialette~\cite{ardevol2025unitmultiple}
%prove the \textsf{NP}-completeness of \cutir\ by a reduction from a
%variant of {\sc 3SAT} 
%as mentioned in \Cref{prop:3sat}, the hardness of the same follows 
%from standard reduction from {\sc 3SAT}.

\begin{proposition}{\rm \cite[Lemma~2.1]{DBLP:journals/algorithmica/FellowsKMP95}}
    \label{prop:3sat}
    The problem of deciding the satisfiability of boolean formulas in conjunctive 
    normal form satisfying the following properties is \npc. Further, the problem cannot 
    be solved in time $2^{o(n+m)}$, assuming \ETH, where $n$ is the number of variables
    and $m$ is the number of clauses in the input formula.
    \begin{itemize}
        \item every clause has either two or three literals over distinct variables,
        \item every variable appears positively exactly once in a 3-clause,
positively exactly once in a 2-clause, and negatively exactly once in a
2-clause,
        \item every clause with
three literals contains only positive literals.
    \end{itemize}
\end{proposition}

% We recall the part of their construction that we need. \sbj{Is there any need of  describing their construction?}
% \srb{Yes. The restrictions in the next Proposition is not explicitly mentioned in the paper. The construction justifies the same.}
% \sbj{For each
% $2$-clause $C_\alpha=(x_i^r\vee x_j^s)$. what it means?}
% \srb{corrected}
\ifthenelse{\boolean{shortver}}{Their reduction proves the hardness for instances satisfying the property stated in \Cref{prop:cutir}. }{
For each variable
$x_i$, the constructed graph contains six vertices
$A_i,B_i,C_i,x_i^1,x_i^2,x_i^N$, where $A_i,B_i,C_i$ are $\black$ and
$x_i^1,x_i^2,x_i^N$ are $\white$. The sets $\{A_i,B_i,C_i\}$ and
$\{x_i^1,x_i^2,x_i^N\}$ are complete to each other. Furthermore,
$A_iC_i$, $B_iC_i$, and $x_i^1x_i^2$ are edges, while $A_iB_i$ is not an
edge. The $\black$ vertices $A_i,B_i,C_i$ have no adjacency outside this
variable gadget. The vertices $x_i^1,x_i^2,x_i^N$ correspond to the three
occurrences of the variable $x_i$, namely the positive occurrence in a
$3$-clause, the positive occurrence in a $2$-clause, and the negative
occurrence in a $2$-clause, respectively. These are the only $\white$
vertices in the constructed graph.

For each $3$-clause $C_\alpha=(x_i\vee x_j\vee x_k)$, the edges
$x_i^1x_j^1$, $x_j^1x_k^1$, and $x_k^1x_i^1$ are added. For each
$2$-clause $C_\alpha$, %where $r,s\in\{2,N\}$, 
two new
$\black$ vertices $L_{i,j}^{\alpha}$ and $p_{i,j}^{\alpha}$, and the edge $p_{i,j}^{\alpha}L_{i,j}^{\alpha}$ are introduced.
Additionally, if $C_\alpha = (x_i\vee x_j)$, then the edges $x_i^2 x_j^2$, $x_i^2 L_{i,j}^{\alpha}$,
$x_j^2 L_{i,j}^{\alpha}$ are introduced; if $C_\alpha = (\overline{x_i}\vee x_j)$, then  $x_i^N x_j^2$, $x_i^N L_{i,j}^{\alpha}$,
$x_j^2 L_{i,j}^{\alpha}$ are introduced; if $C_\alpha = ({x_i}\vee \overline{x_j})$, then  $x_i^2 x_j^N$, $x_i^2 L_{i,j}^{\alpha}$,
$x_j^N L_{i,j}^{\alpha}$ are introduced; and if $C_\alpha = (\overline{x_i}\vee \overline{x_j})$, then  $x_i^N x_j^N$, $x_i^N L_{i,j}^{\alpha}$,
$x_j^N L_{i,j}^{\alpha}$ are introduced.
%and the edges $x_i^r x_j^s$, $x_i^r L_{i,j}^{\alpha}$,
%$x_j^s L_{i,j}^{\alpha}$, and
%$p_{i,j}^{\alpha}L_{i,j}^{\alpha}$ are added, where $r,s\in \{2,N\}$.

It follows from the construction that every $\white$ vertex is adjacent to
two non-adjacent $\black$ vertices, namely the corresponding vertices $A_i$
and $B_i$. Moreover, every $\white$ vertex has a $\white$ neighbor outside its
variable gadget: if the vertex corresponds to a literal in a $3$-clause,
then this neighbor is one of the other two literal vertices in the
triangle added for that clause; and if the vertex corresponds to a
literal in a $2$-clause, then this neighbor is the other literal vertex
of that clause. Since $A_i$ and $B_i$ have no adjacency outside their
variable gadget, this $\white$ neighbor is non-adjacent to both $A_i$ and
$B_i$. Hence every $\white$ vertex is the center of an induced claw
containing two $\black$ leaves and one $\white$ leaf. The \ETH-based lower bound
in \Cref{prop:cutir} follows from the fact that the number of vertices plus edges 
in the constructed graph is linear in the number of variables plus number of clauses in the input
formula.
}

\begin{proposition}{\rm \cite[Theorem~3 and its proof]{ardevol2025unitmultiple}}
\label{prop:cutir}
{\sc \cutir} is \npc\ even for graphs $G$ in which, for every $\white$ vertex
$v$, there exist three mutually non-adjacent neighbors of $v$, two of
which are $\black$ and one of which is $\white$. Moreover, the problem cannot be 
solved in time $2^{o(n+m)}$, unless \ETH fails.
\end{proposition}

To prove the hardness for \uivs, we reduce from the restricted version of \cutir\ in \Cref{prop:cutir}.
Let \((G,\gamma)\) be an instance, and let \(B\) and \(W\) be the sets of
$\black$ and $\white$ vertices, respectively.  Set \(k:=|W|\).  We construct
\(H\) by replacing every $\black$ vertex \(b\in B\) by a clique \(C_b\) of
size \(k+1\), while keeping every $\white$ vertex as a single vertex.  Every
edge of \(G\) is expanded in the natural way: $\black$-$\black$ edges become
complete bipartite graphs between the corresponding cliques, $\black$-$\white$
edges make the $\white$ vertex complete to the corresponding clique, and
$\white$-$\white$ edges are preserved.

The intuition is that the large clique \(C_b\) prevents any $\black$ vertex
from being split: with only \(k\) splits, at least
one vertex in each \(C_b\) remains unsplit.  On the other hand, every
$\white$ vertex is forced to be split.  Indeed, by the restriction in
\Cref{prop:cutir}, each $\white$ vertex is the center of an induced claw
with two $\black$ leaves and one $\white$ leaf; if such a $\white$ vertex were not
split, then the unsplit representatives chosen from the two corresponding
$\black$ cliques, together with a representative of the $\white$ leaf, would
induce a claw.  This is impossible in a unit interval graph.
Thus a solution for the constructed instance splits precisely the $\white$
vertices, and this corresponds to a colored unit \(2\)-interval
split of \((G,\gamma)\), as characterized in \Cref{prop:ui}.  Hence
\((G,\gamma)\) is a \yes-instance of \cutir\ if and only if \((H,k)\) is a
\yes-instance of the vertex-splitting instance. 

\ifthenelse{\boolean{shortver}}{}{Now we give the formal proof.}

\theoremui*

\ifthenelse{\boolean{shortver}}{}{
\begin{proof}
We prove the hardness of \uivs. The hardness of \uisvs follows 
from the fact that all the splits are shallow.
We reduce from the restricted version of \cutir\ stated in
\Cref{prop:cutir}. Let $(G,\gamma)$ be an instance of this restricted
problem. Thus every $\white$ vertex of $G$ is the center of an induced claw
with two $\black$ leaves and one $\white$ leaf.

Let $B=\{v\in V(G):\gamma(v)=\mathsf{black}\}$ and
$W=\{v\in V(G):\gamma(v)=\mathsf{white}\}$. Set $k:=|W|$ and $M:=k+1$.
We construct a graph $H$ as follows. Every $\black$ vertex $b\in B$ is
replaced by a clique $C_b=\{b^1,\ldots,b^M\}$ of size $M$, and every
$\white$ vertex $w\in W$ is kept as a single vertex of $H$. For every edge
of $G$, we add the corresponding complete set of edges in $H$:
\begin{itemize}
    \item if $b,c\in B$ and $bc\in E(G)$, then we make $C_b$ complete
    to $C_c$;
    \item if $b\in B$, $u\in W$, and $bu\in E(G)$, then we make $u$
    adjacent to every vertex of $C_b$;
    \item if $u,v\in W$ and $uv\in E(G)$, then we add the edge $uv$ to
    $H$.
\end{itemize}
For non-edges of $G$, we add no edges between the corresponding vertices
or cliques. The constructed instance of \uivs is $(H,k)$. The construction is clearly polynomial.

We prove that $(G,\gamma)$ is a \yes-instance of \cutir\ if and only if
$(H,k)$ is a \yes-instance of \uivs.

 $(\Rightarrow).$ First suppose that $(G,\gamma)$ is a \yes-instance of \cutir. By
\Cref{prop:ui}, there exists a unit interval graph $S$ and a function
$f:V(S)\rightarrow V(G)$ such that $(S,f)$ is a split of $(G,\gamma)$.
In particular, every $\black$ vertex of $G$ has exactly one representative
in $S$, and every $\white$ vertex of $G$ has exactly two representatives in
$S$.

We construct a graph $S_H$ from $S$ as follows. For every $\black$ vertex
$b\in B$, let $s_b$ be the unique vertex of $S$ such that $f(s_b)=b$.
Replace $s_b$ by a clique of $M$ true twins $s_b^1,\ldots,s_b^M$. The
vertex $s_b^j$ will correspond to the vertex $b^j$ of the clique $C_b$
in $H$. For every $\white$ vertex $w\in W$, keep the two representatives of
$w$ exactly as they are in $S$.

The graph $S_H$ is a unit interval graph. Indeed, unit interval graphs
are closed under adding true twins: in a unit interval representation, a
true twin can be represented by the same unit interval as the original
vertex.

We claim that $S_H$ is obtainable from $H$ without splitting any vertex
belonging to a clique $C_b$, and by splitting each vertex of $W$ exactly
once. Define a function $g:V(S_H)\rightarrow V(H)$ as follows. For every
$\black$ vertex $b\in B$ and every $j\in\{1,\ldots,M\}$, the vertex
$s_b^j$ is mapped to $b^j$, and the two representatives of a $\white$ vertex
$w\in W$ are mapped to $w$.

Each fiber   of $g$ is independent: fibers over vertices of the cliques
$C_b$ have size one, while the fibers over $\white$ vertices have size two
and come from the split $(S,f)$. Moreover, every edge of $S_H$ maps to an
edge of $H$. Indeed, the only new edges in $S_H$ are between true twins
corresponding to vertices inside the same clique $C_b$, or between true
twins corresponding to adjacent vertices of $G$; both types correspond to
edges of $H$.

Conversely, every edge of $H$ is represented in $S_H$. If
$b^j,c^\ell$ are adjacent in $H$, then either $b=c$, in which case
$s_b^j$ and $s_b^\ell$ are adjacent because they belong to the same
clique of true twins, or $bc\in E(G)$, in which case the unique
representatives of $b$ and $c$ are adjacent in $S$, and hence all their
true twins are adjacent in $S_H$. If $b^j\in C_b$ and $u\in W$ with
$b^j u\in E(H)$, then $bu\in E(G)$, and so in $S$ the vertex $s_b$ is
adjacent to some representative of $u$; in $S_H$, the true twin $s_b^j$
is adjacent to the same representative of $u$. Finally, if $u,v\in W$
and $uv\in E(H)$, then $uv\in E(G)$ and this edge is represented in $S$,
and hence also in $S_H$.

Therefore $S_H$ is a unit interval graph obtainable from $H$ by exactly
$|W|=k$ vertex splits. Hence $(H,k)$ is a \yes-instance of
\uivs.

 $(\Leftarrow).$ Conversely, suppose that $(H,k)$ is a \yes-instance of
\uivs. Let $S$ be a unit interval graph
obtained from $H$ using at most $k$ vertex splits. We use the fact that
unit interval graphs are claw-free.

First we claim that every vertex of $W$ must be split. For every
$b\in B$, the clique $C_b$ has size $M=k+1$. Since the total number of
splits is at most $k$, at least one vertex of $C_b$ is not split. Choose
and fix one such unsplit vertex $\widehat b\in C_b$ for every $b\in B$.

Suppose, for a contradiction, that some vertex $w\in W$ is not split. By
\Cref{prop:cutir}, the vertex $w$ has three mutually non-adjacent
neighbors in $G$, say $a_w,b_w,z_w$, where $a_w$ and $b_w$ are $\black$ and
$z_w$ is $\white$. Let $\widehat a_w\in C_{a_w}$ and
$\widehat b_w\in C_{b_w}$ be the fixed unsplit vertices chosen above.
Since $a_wb_w\notin E(G)$, the vertices $\widehat a_w$ and
$\widehat b_w$ are non-adjacent in $H$, and hence they are non-adjacent
in $S$. Since $w$ is unsplit and is adjacent in $H$ to every vertex of
$C_{a_w}\cup C_{b_w}$, the vertices $\widehat a_w$ and $\widehat b_w$
are both adjacent to $w$ in $S$.

Now consider the $\white$ neighbor $z_w$ of $w$. Since $wz_w\in E(H)$ and
$w$ is unsplit, there exists a representative $z'$ of $z_w$ in $S$ that
is adjacent to $w$. Since $z_wa_w,z_wb_w\notin E(G)$, no representative
of $z_w$ can be adjacent to $\widehat a_w$ or to $\widehat b_w$ in $S$.
Thus $z'$ is non-adjacent to both $\widehat a_w$ and $\widehat b_w$.

Therefore the four vertices $w,\widehat a_w,\widehat b_w,z'$ induce a
claw in $S$ centered at $w$. This contradicts the fact that $S$ is a unit
interval graph, since unit interval graphs are claw-free. Thus every
vertex of $W$ must be split.

There are exactly $k=|W|$ vertices in $W$, and the total split budget is
$k$. Hence every vertex of $W$ is split exactly once, and no other vertex
of $H$ is split. In particular, no vertex in any clique $C_b$ is split.

For every $\black$ vertex $b\in B$, choose an arbitrary vertex
$\widehat b\in C_b$.  Since no vertex in any clique $C_b$ is split,
$\widehat b$ is unsplit.  Let $S'$ be the induced subgraph of $S$ on the
chosen vertex $\widehat b$ for every vertex $b\in B$ and on the two
representatives of every vertex $w\in W$.  Since unit interval graphs are
hereditary, $S'$ is also a unit interval graph.

Define a function $f':V(S')\rightarrow V(G)$ as follows. The chosen
vertex $\widehat b\in C_b$ is mapped to $b$ for every $b\in B$, and the
two representatives of every vertex $w\in W$ are mapped to $w$.

We claim that $(S',f')$ is a split of $(G,\gamma)$. First,
$|(f')^{-1}(b)|=1$ for every $b\in B$, and $|(f')^{-1}(w)|=2$ for every
$w\in W$. Moreover, every fiber of $f'$ is independent, because $S$ was
obtained from $H$ by vertex splitting.

Next, if $xy\in E(S')$, then $f'(x)f'(y)\in E(G)$. Indeed, every edge of
$S$ comes from an edge of $H$, and by the construction of $H$, edges of
$H$ occur only between vertices corresponding to adjacent vertices of
$G$, except for edges inside a clique $C_b$. The latter type cannot
appear in $S'$, because $S'$ contains only one vertex from each clique
$C_b$.

Finally, every edge of $G$ is represented in $S'$. Let $uv\in E(G)$. If
$u,v\in B$, then $C_u$ is complete to $C_v$ in $H$, and since
$\widehat u$ and $\widehat v$ are unsplit, the edge
$\widehat u\widehat v$ belongs to $S'$. If $u\in B$ and $v\in W$, then
in $H$ the vertex $v$ is adjacent to every vertex of $C_u$, in particular
to $\widehat u$. Since $\widehat u$ is unsplit, some representative of
$v$ is adjacent to $\widehat u$ in $S'$. If $u,v\in W$, then $uv$ is an
edge of $H$, and so some representative of $u$ is adjacent in $S$ to some
representative of $v$; both representatives belong to $S'$.

Thus $(S',f')$ is a split of $(G,\gamma)$ into a unit interval graph. By
\Cref{prop:ui}, $(G,\gamma)$ is a \yes-instance of \cutir.

We have proved that $(G,\gamma)$ is a \yes-instance of \cutir\ if and only
if $(H,k)$ is a \yes-instance of \uivs.
Now the statements follow from \Cref{prop:cutir} and the facts that
$|V(H)| = \mathcal{O}((|V(G)|+|E(G)|)^2)$, and that $k = \mathcal{O}(|V(G)|+|E(G)|)$.
\end{proof}
}
\paragraph*{Declaration on the use of AI.}
% We used ChatGPT 5.5 in preparing this manuscript. In particular,
% the reductions for \chovs and \uivs were suggested by the model. For
% \chovs, the final reduction was obtained after more than 30 prompts from the
% authors, during which several unsuccessful reductions were explored. For \uivs, 
% the model provided the reduction in as few as five prompts and without
% any hints from the authors; the authors' contribution was limited to verifying
% the correctness of the reduction and polishing the proof.

An AI tool (ChatGPT 5.5) was used to proofread the paper, polish  proofs, and  produce TikZ code for the figures. All such use was closely supervised by authors.  For  \uivs the reduction was suggested by the model; the authors' contribution was limited to verifying
the correctness of the reduction and polishing the proof.
%All other  technical contributions were developed independently of generative AI.
\bibliography{01biblio.bib}

\begin{thebibliography}{10}

\bibitem{abu2019cluster}
Faisal~N. Abu-Khzam, Emmanuel Arrighi, Matthias Bentert, Pål~Grønås Drange, Judith Egan, Serge Gaspers, Alexis Shaw, Peter Shaw, Blair~D. Sullivan, and Petra Wolf.
\newblock {Cluster Editing with Vertex Splitting}.
\newblock {\em Discrete Applied Mathematics}, 371:185--195, 2025.
\newblock \href {https://doi.org/10.1016/j.dam.2025.04.013} {\path{doi:10.1016/j.dam.2025.04.013}}.

\bibitem{DBLP:conf/ai4i/Abu-KhzamBFS21}
Faisal~N. Abu{-}Khzam, Joseph~R. Barr, Amin Fakhereldine, and Peter Shaw.
\newblock {A Greedy Heuristic for Cluster Editing with Vertex Splitting}.
\newblock In {\em 4th International Conference on Artificial Intelligence for Industries, {AI4I} 2021, Laguna Hills, CA, USA, September 20-22, 2021}, pages 38--41. {IEEE}, 2021.
\newblock \href {https://doi.org/10.1109/AI4I51902.2021.00017} {\path{doi:10.1109/AI4I51902.2021.00017}}.

\bibitem{abu2026complexity}
Faisal~N. Abu{-}Khzam, Dipayan Chakraborty, Lucas Isenmann, and Nacim Oijid.
\newblock {On the Complexity of Vertex-Splitting into an Interval Graph}.
\newblock In Florent Foucaud and Aline Parreau, editors, {\em Combinatorial Algorithms - 37th International Workshop, {IWOCA} 2026, Clermont-Ferrand, France, June 8-11, 2026, Proceedings}, volume 16587 of {\em Lecture Notes in Computer Science}, pages 1--15. Springer, 2026.
\newblock \href {https://doi.org/10.1007/978-3-032-27732-9\_1} {\path{doi:10.1007/978-3-032-27732-9\_1}}.

\bibitem{DBLP:conf/cocoon/AbuKhzamDIT25}
Faisal~N. Abu{-}Khzam, Tom Davot, Lucas Isenmann, and Sergio Thoumi.
\newblock {On the Complexity of 2-Club Cluster Editing with Vertex Splitting}.
\newblock In Fedor~V. Fomin and Mingyu Xiao, editors, {\em Computing and Combinatorics - 31st International Computing and Combinatorics Conference, {COCOON} 2025, Chengdu, China, August 15-17, 2025, Proceedings, Part {II}}, volume 15984 of {\em Lecture Notes in Computer Science}, pages 3--14. Springer, 2025.
\newblock \href {https://doi.org/10.1007/978-981-95-0218-9\_1} {\path{doi:10.1007/978-981-95-0218-9\_1}}.

\bibitem{DBLP:journals/corr/abs-1901-00156}
Faisal~N. Abu{-}Khzam, Judith Egan, Serge Gaspers, Alexis Shaw, and Peter Shaw.
\newblock {Cluster Editing with Vertex Splitting}.
\newblock In Jon Lee, Giovanni Rinaldi, and Ali~Ridha Mahjoub, editors, {\em Combinatorial Optimization - 5th International Symposium, {ISCO} 2018, Marrakesh, Morocco, April 11-13, 2018, Revised Selected Papers}, volume 10856 of {\em Lecture Notes in Computer Science}, pages 1--13. Springer, 2018.
\newblock \href {https://doi.org/10.1007/978-3-319-96151-4\_1} {\path{doi:10.1007/978-3-319-96151-4\_1}}.

\bibitem{DBLP:conf/iwoca/AbuKhzamIM25}
Faisal~N. Abu{-}Khzam, Lucas Isenmann, and Zeina Merchad.
\newblock {Bicluster Editing with Overlaps: {A} Vertex Splitting Approach}.
\newblock In Henning Fernau and Binhai Zhu, editors, {\em Combinatorial Algorithms - 36th International Workshop, {IWOCA} 2025, Bozeman, MT, USA, July 21-24, 2025, Proceedings}, volume 15885 of {\em Lecture Notes in Computer Science}, pages 146--159. Springer, 2025.
\newblock \href {https://doi.org/10.1007/978-3-031-98740-3\_11} {\path{doi:10.1007/978-3-031-98740-3\_11}}.

\bibitem{abu2025complexity}
Faisal~N. Abu{-}Khzam and Sergio Thoumi.
\newblock {On the Complexity of Claw-Free Vertex Splitting}.
\newblock {\em CoRR}, abs/2506.06044, 2025.
\newblock \href {https://doi.org/10.48550/ARXIV.2506.06044} {\path{doi:10.48550/ARXIV.2506.06044}}.

\bibitem{DBLP:conf/gd/AhmedKK22}
Abu~Reyan Ahmed, Stephen~G. Kobourov, and Myroslav Kryven.
\newblock {An {FPT} Algorithm for Bipartite Vertex Splitting}.
\newblock In Patrizio Angelini and Reinhard von Hanxleden, editors, {\em Graph Drawing and Network Visualization - 30th International Symposium, {GD} 2022, Tokyo, Japan, September 13-16, 2022, Revised Selected Papers}, volume 13764 of {\em Lecture Notes in Computer Science}, pages 261--268. Springer, 2022.
\newblock \href {https://doi.org/10.1007/978-3-031-22203-0\_19} {\path{doi:10.1007/978-3-031-22203-0\_19}}.

\bibitem{baumann2023parameterized}
Jakob Baumann, Matthias Pfretzschner, and Ignaz Rutter.
\newblock {Parameterized Complexity of Vertex Splitting to Pathwidth at Most 1}.
\newblock In {\em International Workshop on Graph-Theoretic Concepts in Computer Science}, pages 30--43. Springer, 2023.
\newblock \href {https://doi.org/10.1007/978-3-031-43380-1\_3} {\path{doi:10.1007/978-3-031-43380-1\_3}}.

\bibitem{DBLP:journals/tcs/BaumannPR24}
Jakob Baumann, Matthias Pfretzschner, and Ignaz Rutter.
\newblock Parameterized complexity of vertex splitting to pathwidth at most 1.
\newblock {\em Theor. Comput. Sci.}, 1021:114928, 2024.
\newblock \href {https://doi.org/10.1016/J.TCS.2024.114928} {\path{doi:10.1016/J.TCS.2024.114928}}.

\bibitem{chlebik2006approximation}
Miroslav Chleb{\'{\i}}k and Janka Chleb{\'{\i}}kov{\'{a}}.
\newblock Approximation hardness of edge dominating set problems.
\newblock {\em J. Comb. Optim.}, 11(3):279--290, 2006.
\newblock \href {https://doi.org/10.1007/S10878-006-7908-0} {\path{doi:10.1007/S10878-006-7908-0}}.

\bibitem{chlebik2007complexity}
Miroslav Chleb{\'{\i}}k and Janka Chleb{\'{\i}}kov{\'{a}}.
\newblock {The Complexity of Combinatorial Optimization Problems on d-Dimensional Boxes}.
\newblock {\em {SIAM} J. Discret. Math.}, 21(1):158--169, 2007.
\newblock \href {https://doi.org/10.1137/050629276} {\path{doi:10.1137/050629276}}.

\bibitem{CyganFKLMPPS15}
Marek Cygan, Fedor~V. Fomin, Lukasz Kowalik, Daniel Lokshtanov, D{\'{a}}niel Marx, Marcin Pilipczuk, Michal Pilipczuk, and Saket Saurabh.
\newblock {\em {Parameterized Algorithms}}.
\newblock Springer, 2015.
\newblock \href {https://doi.org/10.1007/978-3-319-21275-3} {\path{doi:10.1007/978-3-319-21275-3}}.

\bibitem{Diestel17}
Reinhard Diestel.
\newblock {\em {Graph Theory, 5th Edition}}, volume 173 of {\em Graduate texts in mathematics}.
\newblock Springer, 2017.
\newblock \href {https://doi.org/10.1007/978-3-662-53622-3} {\path{doi:10.1007/978-3-662-53622-3}}.

\bibitem{eades1995vertex}
Peter Eades and CFX de~Mendon{\c{c}}a~N.
\newblock Vertex splitting and tension-free layout.
\newblock In {\em International Symposium on Graph Drawing}, pages 202--211. Springer, 1995.
\newblock \href {https://doi.org/10.1007/BFB0021804} {\path{doi:10.1007/BFB0021804}}.

\bibitem{faria2001splitting}
Lu{\'e}rbio Faria, Celina~MH de~Figueiredo, and CFX de~Mendon{\c{c}}a~N.
\newblock Splitting number is {NP}-complete.
\newblock {\em Discrete Applied Mathematics}, 108(1-2):65--83, 2001.
\newblock \href {https://doi.org/10.1016/S0166-218X(00)00220-1} {\path{doi:10.1016/S0166-218X(00)00220-1}}.

\bibitem{DBLP:journals/algorithmica/FellowsKMP95}
Michael~R. Fellows, Jan Kratochv{\'{\i}}l, Matthias Middendorf, and Frank Pfeiffer.
\newblock {The Complexity of Induced Minors and Related Problems}.
\newblock {\em Algorithmica}, 13(3):266--282, 1995.
\newblock \href {https://doi.org/10.1007/BF01190507} {\path{doi:10.1007/BF01190507}}.

\bibitem{firbas2023establishing}
Alexander Firbas.
\newblock {\em Establishing hereditary graph properties via vertex splitting}.
\newblock PhD thesis, Technische Universit{\"a}t Wien, 2023.
\newblock \href {https://doi.org/10.34726/hss.2023.103864} {\path{doi:10.34726/hss.2023.103864}}.

\bibitem{DBLP:conf/isaac/FirbasS24}
Alexander Firbas and Manuel Sorge.
\newblock {On the Complexity of Establishing Hereditary Graph Properties via Vertex Splitting}.
\newblock In Juli{\'{a}}n Mestre and Anthony Wirth, editors, {\em 35th International Symposium on Algorithms and Computation, {ISAAC} 2024, Sydney, Australia, December 8-11, 2024}, volume 322 of {\em LIPIcs}, pages 30:1--30:15. Schloss Dagstuhl - Leibniz-Zentrum f{\"{u}}r Informatik, 2024.
\newblock \href {https://doi.org/10.4230/LIPICS.ISAAC.2024.30} {\path{doi:10.4230/LIPICS.ISAAC.2024.30}}.

\bibitem{gaikwad2025inclusive}
Ajinkya Gaikwad, Hitendra Kumar, S~Padmapriya, Praneet~Kumar Patra, Harsh Sanklecha, and Soumen Maity.
\newblock {Inclusive and Exclusive Vertex Splitting into Specific Graph Classes: {NP} Hardness and Algorithms}.
\newblock {\em arXiv preprint arXiv:2510.26938}, 2025.

\bibitem{hartsfield1985splitting}
Nora Hartsfield, Brad Jackson, and Gerhard Ringel.
\newblock The splitting number of the complete graph.
\newblock {\em Graphs and Combinatorics}, 1(1):311--329, 1985.
\newblock \href {https://doi.org/10.1007/BF02582960} {\path{doi:10.1007/BF02582960}}.

\bibitem{y2008improving}
Nathalie Henry, Anastasia Bezerianos, and Jean{-}Daniel Fekete.
\newblock {Improving the Readability of Clustered Social Networks using Node Duplication}.
\newblock {\em {IEEE} Trans. Vis. Comput. Graph.}, 14(6):1317--1324, 2008.
\newblock \href {https://doi.org/10.1109/TVCG.2008.141} {\path{doi:10.1109/TVCG.2008.141}}.

\bibitem{hilton1997vertex}
Anthony~JW Hilton and C~Zhao.
\newblock Vertex-splitting and chromatic index critical graphs.
\newblock {\em Discrete applied mathematics}, 76(1-3):205--211, 1997.
\newblock \href {https://doi.org/10.1016/S0166-218X(96)00125-4} {\path{doi:10.1016/S0166-218X(96)00125-4}}.

\bibitem{DBLP:journals/jcss/ImpagliazzoP01}
Russell Impagliazzo and Ramamohan Paturi.
\newblock {On the Complexity of k-SAT}.
\newblock {\em J. Comput. Syst. Sci.}, 62(2):367--375, 2001.
\newblock \href {https://doi.org/10.1006/JCSS.2000.1727} {\path{doi:10.1006/JCSS.2000.1727}}.

\bibitem{jackson1984splitting}
Brad Jackson and Gerhard Ringel.
\newblock The splitting number of complete bipartite graphs.
\newblock {\em Archiv der Mathematik}, 42(2):178--184, 1984.

\bibitem{jackson1985splittings}
Brad Jackson and Gerhard Ringel.
\newblock {Splittings of Graphs on Surfaces}.
\newblock In Frank Harary, editor, {\em Graphs and Applications}, Proceedings of the 1st Colorado Symposium on Graph Theory, pages 203--219. Wiley, Boulder, Colorado, 1985.

\bibitem{komusiewicz2018tight}
Christian Komusiewicz.
\newblock {Tight Running Time Lower Bounds for Vertex Deletion Problems}.
\newblock {\em {ACM} Trans. Comput. Theory}, 10(2):6:1--6:18, 2018.
\newblock \href {https://doi.org/10.1145/3186589} {\path{doi:10.1145/3186589}}.

\bibitem{kumar2026parameterized}
Hitendra Kumar.
\newblock {\em {The Parameterized Complexity of Graph Editing \& MaxMin Problems}}.
\newblock PhD thesis, IISER Pune, 2026.

\bibitem{ardevol2025unitmultiple}
Virginia~Ard{\'e}vol Mart{\'i}nez, Romeo Rizzi, Florian Sikora, and St{\'e}phane Vialette.
\newblock Recognizing unit multiple interval graphs is hard.
\newblock {\em Discrete Applied Mathematics}, 360:258--274, 2025.
\newblock \href {https://doi.org/10.1016/j.dam.2024.09.011} {\path{doi:10.1016/j.dam.2024.09.011}}.

\bibitem{nollenburg2025planarizing}
Martin N{\"o}llenburg, Manuel Sorge, Soeren Terziadis, Ana{\"\i}s Villedieu, Hsiang-Yun Wu, and Jules Wulms.
\newblock Planarizing graphs and their drawings by vertex splitting.
\newblock {\em Journal of Computational Geometry}, 16(1):333--372, 2025.
\newblock \href {https://doi.org/10.20382/JOCG.V16I1A10} {\path{doi:10.20382/JOCG.V16I1A10}}.

\bibitem{Rose70}
Donald~J. Rose.
\newblock Triangulated graphs and the elimination process.
\newblock {\em Journal of Mathematical Analysis and Applications}, 32(3):597--609, 1970.
\newblock \href {https://doi.org/10.1016/0022-247X(70)90282-9} {\path{doi:10.1016/0022-247X(70)90282-9}}.

\bibitem{RoseTL76}
Donald~J. Rose, Robert~Endre Tarjan, and George~S. Lueker.
\newblock {Algorithmic Aspects of Vertex Elimination on Graphs}.
\newblock {\em {SIAM} J. Comput.}, 5(2):266--283, 1976.
\newblock \href {https://doi.org/10.1137/0205021} {\path{doi:10.1137/0205021}}.

\bibitem{sandeep2015parameterized}
RB~Sandeep and Naveen Sivadasan.
\newblock Parameterized lower bound and improved kernel for diamond-free edge deletion.
\newblock In {\em 10th International Symposium on Parameterized and Exact Computation (IPEC 2015)}, pages 365--376. Schloss Dagstuhl--Leibniz-Zentrum f{\"u}r Informatik, 2015.
\newblock \href {https://doi.org/10.4230/LIPIcs.IPEC.2015.365} {\path{doi:10.4230/LIPIcs.IPEC.2015.365}}.

\bibitem{yannakakis1981computing}
Mihalis Yannakakis.
\newblock {Computing the minimum fill-in is NP-complete}.
\newblock {\em SIAM Journal on Algebraic Discrete Methods}, 2(1):77--79, 1981.

\bibitem{DBLP:journals/siamcomp/Yannakakis81a}
Mihalis Yannakakis.
\newblock {Node-Deletion Problems on Bipartite Graphs}.
\newblock {\em {SIAM} J. Comput.}, 10(2):310--327, 1981.
\newblock \href {https://doi.org/10.1137/0210022} {\path{doi:10.1137/0210022}}.

\end{thebibliography}

\end{document}